\documentclass{article}
\RequirePackage{hyperref}
\RequirePackage{amsthm,amsmath,amsbsy,amssymb,bm,color, graphics, epsfig, epstopdf,url, makeidx,tocbibind,  textcomp}
\usepackage{algorithm}
\usepackage[noend]{algpseudocode}

\RequirePackage[acronym]{glossaries}
\usepackage[utf8]{inputenc}
\usepackage{geometry}
\usepackage[round, comma]{natbib}
\newtheorem{definition}{Definition}

\newtheorem{proposition}{Proposition}
\newtheorem{theorem}{Theorem}
\newtheorem{lemma}{Lemma}

\newcommand{\E}{\ensuremath{\mathrm{E}}}
\newcommand{\var}{\ensuremath{\mathrm{var}}}
\newcommand{\cov}{\ensuremath{\mathrm{cov}}}
\newcommand{\R}{\mathbb{R}}
\newcommand{\C}{\mathbb{C}}
\newcommand{\N}{\mathbb{N}}
\newcommand{\Z}{\mathbb{Z}}
\newcommand{\btheta}{\boldsymbol{\theta}}

\newcommand{\bgamma}{\boldsymbol{\gamma}}
\newcommand{\sN}{^{(N)}}
\newcommand{\sfN}{^{(\phi(N))}}
\newcommand{\ufN}{_{\phi(N)}}

\newcommand{\inerf}{\omega^{\{f\}}}
\newcommand{\Y}{\widetilde{X}}
\makeindex

\usepackage[affil-it]{authblk}
\makeatletter
\def\@maketitle{%
  \newpage
  \null
  \vskip 2em%
  \begin{center}%
  \let \footnote \thanks
    {\Large\bfseries \@title \par}%
    \vskip 1.5em%
    {\normalsize
      \lineskip .5em%
      \begin{tabular}[t]{c}%
        \@author
      \end{tabular}\par}%
    \vskip 1em%
    {\normalsize \@date}%
  \end{center}%
  \par
  \vskip 1.5em}
\makeatother

\begin{document}
\title{Analysis Of Nonstationary Modulated Time Series With Applications to Oceanographic Surface Flow Measurements}
\author[1]{Arthur P. Guillaumin\thanks{Corresponding author: \texttt{arthur.guillaumin.14@ucl.ac.uk}}}%
\author{Adam M. Sykulski}
\author[1]{Sofia C. Olhede}
\affil[1]{Department of Statistical Science\\ University College London, UK}
\author{Jeffrey J. Early}
\author{Jonathan M. Lilly}
\affil{NorthWest Research Associates, Seattle WA, USA}
\date{}
\maketitle
\begin{abstract} 
We propose a new class of univariate nonstationary time series models, using the framework of modulated time series, which is appropriate for the analysis of rapidly-evolving time series as well as time series observations with missing data. We extend our techniques to a class of bivariate time series that are isotropic. Exact inference is often not computationally viable for time series analysis, and so we propose an estimation method based on the Whittle-likelihood, a commonly adopted pseudo-likelihood. Our inference procedure is shown to be consistent under standard assumptions, as well as having considerably lower computational cost than exact likelihood in general. We show the utility of this framework for the analysis of drifting instruments, an analysis that is key to characterising global ocean circulation and therefore also for decadal to century-scale climate understanding.
\end{abstract}

\section{Introduction}
This paper introduces a new family of rapidly-evolving time series models, inspired by real data applications, and then develops the appropriate analysis tools for their computationally-efficient and consistent inference. Statistical models for time series observations are usually described by their expectations and covariance structure. Classic families of covariance structure correspond to stationary covariances, governed only by the temporal lags between observed values of the process. The assumption of stationarity greatly simplifies analysis, as it renders the covariance structure homogeneous across time and this motivates averaging for estimation. Unfortunately most often this homogeneous time structure is inadequate as a model for real-world applications, and does not reflect the variability of the observed time series.

In order to analyse nonstationary time series, using the framework of locally stationary time series is standard \citep{priestley1988non, Dahlhaus1997}. The idea is to allow for a time-varying spectral density. Parametric models for the time-varying spectral density can be fitted via the use of local Fourier transforms, usually requiring a spectral smoothness assumption. The concept of infill asymptotics developed by \citet{Dahlhaus1997} is based on the idea that a growing amount of data is obtained locally in time.
%This notion of a limit makes sense when we expect to see little variation over a time analysis window. \textcolor{red}{For our type of time series, this will not be realistic, as here we expect to be in the regime where as we sample more finely additional complexity is constructed}.
Normally, for nonstationary time series analysis, there is a bias-variance trade-off that occurs when selecting the length of an analysis window. Longer windows will decrease variance, but will simultaneously increase bias due to the variation of the covariance function over the analysis window \citep{adak1998time}. In our case we shall eliminate the bias, and this will enable us to use longer time window lengths. 
%To remove the bias, we must first understand it, and therefore turn to the class of modulated processes introduced by Parzen and Priestley.
%We need a model class that allows us to understand the bias/variance trade-off due to collecting more data in time, where the bias is expected to come from the nonstationarity and the variance is expected to come both from nonstationarity and finite sampling. 
For this purpose we exploit the notion of a modulated process~\citep{Parzen1963, priestley1965evol}. A modulated process is a latent stationary process multiplied pointwise by a modulating function. If we observe the modulating function, this framework allows us to define an averaged autocovariance function, despite the clear nonstationarity of the modulated process. This in turn allows us to introduce the Fourier transform of the averaged autocovariance function of the modulated process, which is equal to its expected periodogram. Through examining the expected periodogram, the properties controlling the latent random process may be inferred even when the modulating function changes very rapidly.

The standard class of modulated processes are asymptotically stationary modulated processes~\citep{Parzen1963, Iacobucci2003, Jiang2004}. Here the autocovariance of the modulating function converges to a fixed function, which is too restrictive for our real-world application. We introduce a more general class, which we call \emph{modulated processes with a significant correlation contribution}. This more flexible model will still allow us to infer the parameters of the driving process using likelihood-based methodologies. An alternative approach might be to simply divide the observed process by the known modulating sequence to recover the latent process,  and then perform inference directly on the recovered latent process. However this is not possible in general, as the modulating function may contain zeros, or the observed process may in fact be an aggregation of different processes, as will be the case in the real-world application that motivated us to develop this model class.

Anticipating our application to oceanographic surface flow measurements, we present a novel generalization of modulated processes for isotropic bivariate processes, or equivalently proper complex-valued processes \citep{complexProcesses}. The wealth of possible structure in multivariate processes
is considerable in general. Inherent documented challenges in modelling include producing valid joint representations~\citep{tong1973some,tong1974time,priestley1973analysis}. This problem does not apply here as we shall modulate both processes under consideration simultaneously, thus automatically removing such problems.

Having set up our model of modulated processes with a significant correlation contribution, we show how a modified version of a frequency-domain likelihood allows us to consistently estimate parameters with a high degree of computational efficiency. More specifically, the Whittle likelihood for stationary Gaussian processes is an approximation of the exact likelihood that is consistent and can be computed in $\mathcal{O}(N\log N)$ elementary operations. 
We adapt this pseudo-likelihood to our class of models, making use of the expected periodogram, and conserve the $\mathcal{O}_P\left(N^{-1/2}\right)$ convergence rate. We also conserve the $\mathcal{O}(N\log N)$ computational cost in the minimization procedure, except for a pre-computational step of $\mathcal{O}(N^2)$ operations, which must be performed only once per observed time series sample. Exact likelihood for nonstationary time series, on the other hand, will in general require more than $\mathcal{O}(N^2)$ operations, due to the need to manipulate large covariance matrices.

We apply this method to an important dataset measuring ocean currents. There are only a handful of observational platforms capable of providing continuous global coverage of the Earth's oceans and so it is critical that we fully utilize these datasets to advance our understanding of the oceans and their impact on climate.
One of these studies is the Global Drifter Program (GDP, www.aoml.noaa.gov/phod/dac), consisting of freely drifting instruments, or ``drifters" \citep{lumpkin07}. %, an example being provided in Fig. 1. Recent studies by \citet{sykulski2016Lagrangian} have shown that the velocity time series of drifters can be summarized using locally stationary stochastic models. However, as the drifters traverse highly heterogeneous regions of the ocean, the assumption of local stationarity is often strongly violated.
%For example, in Fig. 1 we display 200 drifter trajectories from the equatorial regions, which are known to be highly nonstationary regions of the global ocean system. 
Fig. 1(a) shows positions from multiple trajectories obtained from drifters at or near the equator. From the positions of the trajectories, we may also calculate the velocities of the instruments, and these velocity time series are useful measurements for understanding ocean dynamics.
Depending on the instrument it may not be reasonable to model the velocity time series as locally stationary, as is assumed in \citet{sykulski2016Lagrangian}.
In particular, for reasons to be discussed, regions near the equator are likelier to yield drifter trajectories with highly nonstationary velocities where locally stationary modelling breaks down. Instead, to capture such rapid time-variability, we use a modulated stochastic process from our class of nonstationary models. This model allows us to capture the rapid frequency modulation of oscillations known to geophysicists as ``inertial oscillations''. An example of a time series with such rapid frequency modulation can be seen in Fig. 1(c).

We organize the paper into the following sections. Section 2 reviews the model family of modulated processes, the standard assumption of asymptotic stationarity associated with such processes, and introduces our generalized class called modulated processes with a significant correlation contribution. This section also includes extensions to bivariate processes. Section 3 describes a computationally consistent pseudo-likelihood estimation procedure. In Section 4 we apply our methods to real-world oceanographic data and various numerical experiments; we also apply our methods to a simulated missing data problem. 
We establish consistency of our proposed procedure in Section 5, under the assumption of significant correlation contribution, as well as standard assumptions on the stationary process that is modulated. Finally, concluding remarks can be found in Section 6.
\begin{figure}[h!]
\centering
\includegraphics[width=0.8\textwidth]{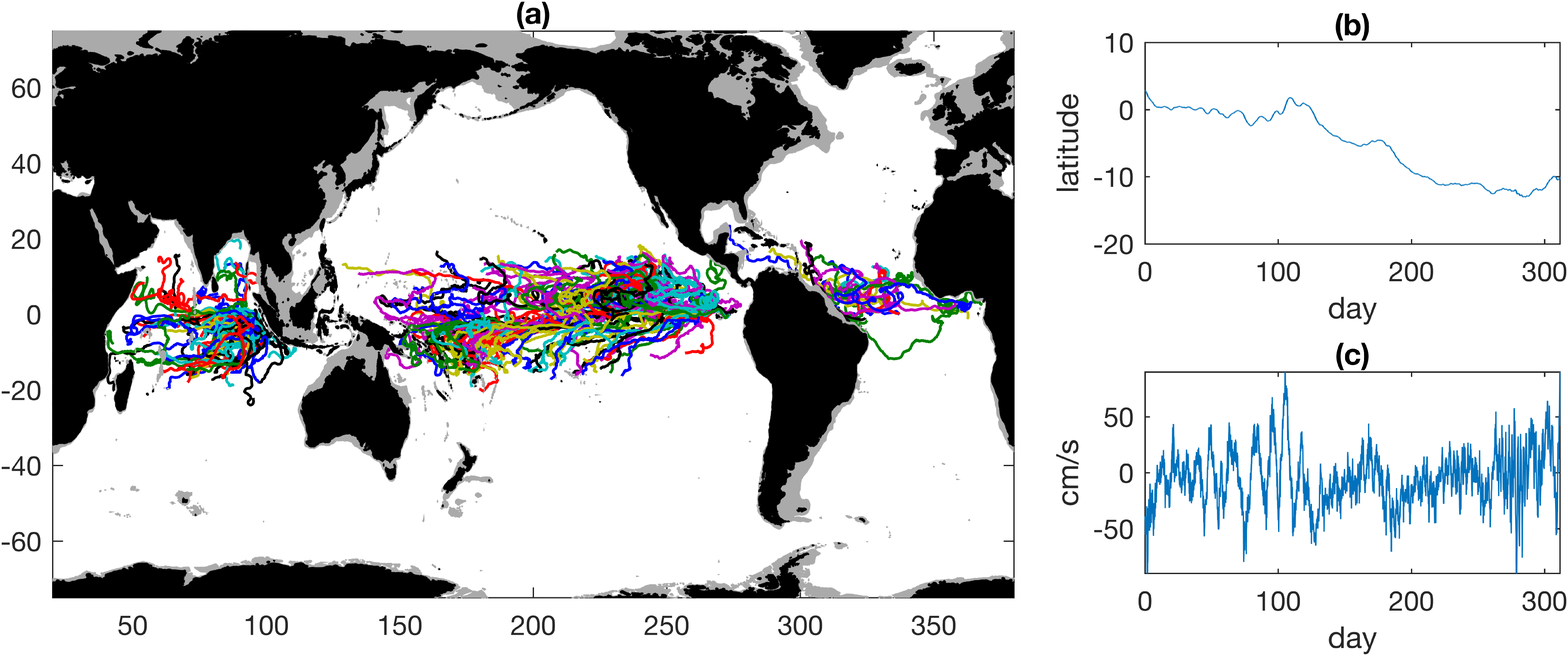}
\caption{(a) The trajectories of the 200 drifters from the Global Drifter Program, analysed in Section~\ref{sec=GDPapp}, that exhibit the greatest change in Coriolis frequency ($f$) across 60 inertial cycles, as described in that section; (b) a segment of data of the meridional (latitudinal) positions over time from Drifter ID\#43594; and (c) a segment of data of the meridional velocities from this drifter in cm/s. This figure is produced using the jLab toolbox~\citep{jlab}.}
\label{Equatorial}
\end{figure}

\section{Modulated time series}
In this paper we review and study modelling and inference methods for univariate and bivariate nonstationary Gaussian processes.
We have that the first moment of a univariate stochastic discrete Gaussian process $\{X_t:t\in\N\}$, with index set $\N = \left\{0,1,2,\cdots \right\}$, is provided pointwise by
\begin{equation*}
\mu_X(t)=\E \left\{X_t\right\},
\end{equation*}
and the second-order structure is given by
\begin{equation*}
c_X(t_1,t_2)=\cov\left\{ X_{t_1},X_{t_2}\right\},
\end{equation*}
where moments are finite as a direct consequence of the joint Gaussianity of $\{X_t\}$.
We shall assume throughout this paper that $\mu_X(t) =  0$. In practice this may require us to subtract the sample mean from the observed series, or more generally remove trends and seasonal components \citep[see][chap. 1]{BrockwellDavisTimeSeries}.

Second-order stationarity implies that the function $c_X(t,t+\tau)$ does not depend on the index $t$ and takes the simplified form $c_X(\tau)$. An alternative way to represent $c_X(\tau)$, assuming it is absolutely summable, is via its Fourier transform $S_X(\cdot)$, also known as the spectral density of $\{X_t\}$,
\[
S_X(\omega) =\frac{1}{2\pi}\sum_{\tau=-\infty}^\infty{c_X(\tau)e^{-i\omega\tau}}, \ \ \omega\in [-\pi,\pi].
\]
The spectral density $S_X(\omega)$ is then a continuous function of $\omega$. The corresponding inversion formulae is given for all integer value $\tau$ by
\[
c_X(\tau) = \int_{-\pi}^\pi{S_X(\omega)e^{i\omega\tau}d\omega}.
\]
A consequence of stationarity is that the quantities in question can be stably estimated by averaging in time \citep{BrockwellDavisTimeSeries}. If  $c_X(t,t+\tau)$ is not stationary, but is slowly varying in time, then it can be estimated by dividing the observed data into multiple segments and performing inference on each segment~\citep{adak1998time}.
This does not hold in settings where the time variation is too rapid.
Our goal is to estimate $c_X(t,t+\tau)$ in such settings, in particular when a parametric specification is made for the function.

\subsection{Classes of modulated processes}
Modulation is a natural and simple method of producing a nonstationary process~\citep{Parzen1963}.
A univariate modulated process is defined as follows.
\begin{definition}[Modulated process]
\label{def=modulatedprocesses}
Let $\{X_t:t\in\N\}$ be a Gaussian, real-valued, zero-mean stationary process.
Let $\{g_t:\; t\in {\mathbb{N}}\}$ be a given bounded real-valued deterministic sequence. 
Then a modulated process is defined as one taking the form
\begin{equation}
\label{modeqn}
\Y_t=g_t X_t
\end{equation} 
at all time points $t\in\N$. 
\end{definition}
Herein we treat $\{g_t\}$ as a known deterministic signal.
%, although our results may be readily extended to incorporate a random $g_t$ that is observed \textcolor{red}{(i.e known for any given realization of $\{\Y_t\}$) and statistically independent from $X_t$ \citep{Dunsmuir1981b}, in which case we could condition on the observed $g_t$} 
In our setting the process $\{X_t\}$, which is referred to as the {\em latent} process, is modelled through a finite set of parameters $\boldsymbol{\theta}\in\Theta\subset\R^d$, where $d$ is a positive integer and $\Theta$ is the parameter space. Usually our object of interest is $\boldsymbol{\theta}$, the particular values of parameters that generated the observed realization.
For example, if the latent process is an autoregressive process of order $p$, we then have $d = p+1$ if the mean is known ($p$ regressive parameters and the variance of the innovations).
We denote the autocovariance function of the stationary zero-mean process $\{X_t\}$ by $c_X(\tau)$, or $c_X(\tau;\boldsymbol{\theta})$ when we want to make the dependence on $\boldsymbol{\theta}$ explicit. Its Fourier transform, the spectral density, is denoted $S_X(\omega)$ or $S_X(\omega;\btheta)$, respectively.

The modulation of the latent process $X_t$ is a convenient mechanism to account for a wide range of nonstationary processes. In particular this mechanism has been widely used as a modelling tool for missing data problems, where $g_t$ is assigned values $0$ or $1$ when respectively missing or observing a data point in time \citep{R.H}. 
%Another key application is in seasonal multiplicative models, common in financial time series analysis or climate-related studies.

To understand when we can recover the parameters controlling the latent process $X_t$ from observing $\Y_t$, we need to put further conditions in place on $g_t$. 
The time series $\Y_t/g_t$ cannot always be formed as $g_t$ may be zero for some time indices, corresponding to missing observations. Another reason is that we may not directly observe $\Y_t$, but instead we may observe an aggregated process $\Y_t + Z_t$, where $Z_t$ is a stationary process (or more generally another modulated process) independent from $\Y_t$, this preventing us from recovering the stationary latent process $X_t$ by division.

%The concept of modulated time series is not new. It has been introduced in \citet{R.H} where the particular case of regularly missed observations was considered.
%Standard theory developed by \citet{parzen1961,Parzen1963,Dunsmuir1981} corresponds to the following conditions, under which the generating mechanisms of $X_t$ can be recovered. 

We assume that $\Y_t$ satisfies~\eqref{modeqn} for a Gaussian, real-valued, zero-mean stationary $X_t$ with absolutely summable autocovariance sequence. Then $\E\{\Y_t\}=g_t\E\{X_t\}=0$  and the time-varying autocovariance sequence is defined by $c_{\Y}(t, t+\tau;\btheta)=\E\left\{\Y_t\Y_{t+\tau}\right\}$. Given a single length N realization $\Y_0,\cdots,\Y_{N-1}$, we start by computing the usual method of moments estimator according to
\begin{equation}
\label{eq=biasedAutocovEst}
\hat{c}_{\Y}\sN(\tau)=\frac{1}{N}\sum_{t=0}^{N-\tau-1} \Y_t\Y_{t+\tau},
\end{equation}
for $\tau=0,1,...,N-1$, such that $\tau$ is within the range of time offsets that is permissible given the length-$N$ sample.
Equation~\eqref{eq=biasedAutocovEst} is the biased sample autocovariance sequence of the modulated time series, which we define even though the process is nonstationary, as this object will become pivotal in our estimation procedure.
The expectation of this object, which we denote $\overline{c}_{\Y}^{(N)}(\tau;\btheta)$ or simply $\overline{c}_{\Y}^{(N)}(\tau)$, takes the following form,
\begin{equation}
	%\label{eq=def_expected_autocov}
	\label{eq=expectedAutocovSequence1}
	\overline{c}_{\Y}\sN(\tau)= \E\{\hat{c}_{\Y}\sN(\tau)\} = \E\left\{ \frac{1}{N}\sum_{t=0}^{N-\tau-1} \Y_t\Y_{t+\tau} \right\} = c_X(\tau)\frac{1}{N}\sum_{t=0}^{N-\tau-1} g_t g_{t+\tau} = c_g\sN(\tau)\cdot c_X(\tau),
\end{equation}
where we have introduced the (deterministic) sample autocovariance of the modulating sequence,
\begin{equation}
	\label{eq=autocovOfg}
	c_g\sN(\tau) = \frac{1}{N}\sum_{t=0}^{N-1-k}{g_tg_{t+\tau}}.
\end{equation}
In the specific case where the modulating sequence $\{g_t\}$ is constant and equal to unity everywhere, which would correspond to observing the latent stationary process directly, we recover the expectation of the biased sample autocovariance for stationary time series, $\left(1-\tau/N\right) c_X(\tau)$, for $\tau=0,\cdots,N-1$.
More generally, a standard assumption is to say that the modulated process $\Y_t$ is an asymptotically stationary process~\citep{parzen1961,Parzen1963}, which arises if for all lags $\tau$, the quantity $c_g\sN(\tau)$ in~\eqref{eq=expectedAutocovSequence1} converges as $N$ tends to infinity. We define this formally as follows.

%a simple way to understand~\eqref{eq=expectedAutocovSequence1} is by assuming that for all lags $\tau$, the quantity $c_g\sN(\tau)$ converges when $N$ goes to infinity. This corresponds to a standard assumption whereby the modulated process $\Y_t$ is an asymptotically stationary process~\citep{parzen1961,Parzen1963}, a notion that we now formally define.

\begin{definition}[Asymptotically stationary process]\label{asympstat}
Let $\{\Y_t\}$ be a discrete time random process.
We say that
$\{\Y_t\}$ is an asymptotically stationary process if there exists a fixed function $\{\gamma(\tau):\tau\in\N \}$ such that for all $\tau\in\N$,
\begin{equation}
	 \lim_{N\rightarrow \infty}\E\left\{\frac{1}{N}\sum_{t=0}^{N-\tau-1} \Y_t\Y_{t+\tau}\right\} =   \gamma(\tau),
\end{equation}
or specifically if $\Y_t$ is a modulated process as defined in Definition~\ref{def=modulatedprocesses}, $\Y_t$ is asymptotically stationary if,
\begin{equation}
\label{eq=asymptStat}
	\lim_{N\rightarrow \infty} \overline{c}_{\Y}\sN(\tau) = \gamma(\tau),
\end{equation}
where $\overline{c}_{\Y}\sN(\tau)$ is defined in~\eqref{eq=expectedAutocovSequence1}.
\end{definition}
An example of a nonstationary but asymptotically stationary process is given by \citet{Parzen1963}, where a stationary process is observed according to a periodically missing data pattern, such that the first $k$ values are observed, the next $l$ values are missed, the next $k$ values are observed, and so on, where $k$ and $l$ are two strictly positive integers.

The class of asymptotically stationary modulated processes \citep{ Parzen1963, Dunsmuir1981, Iacobucci2003, Jiang2004} corresponds to that for which there exists a sequence $\{R_g(\tau):\tau\in\N\}$ such that
\begin{equation}
\label{asymptotic}
\lim_{N\rightarrow \infty} c_g\sN(\tau) = R_g(\tau), \ \ \forall\tau\in\N.
\end{equation}
Indeed we then note that $\overline{c}_{\Y}\sN(\tau) \rightarrow R_g(\tau) c_X(\tau)$ as $N\longrightarrow\infty$, so we could estimate $c_X(\tau)$ by defining
\begin{equation}
\label{estimate}
	\hat{c}_X\sN(\tau)=\frac{\hat{c}_{\Y}\sN(\tau)}{R_g(\tau)},
\end{equation}
assuming $R_g(\tau) \neq 0$ for all $\tau\in\N$, and is known.
It is shown in \citet{Parzen1963} that $\hat{c}_X\sN(\tau)$ is a consistent estimator of the autocovariance sequence $c_X(\tau)$ of the latent stationary process, under some rather mild conditions. Further results are found in~\cite{Dunsmuir1981}. Consistent spectral density estimates can be obtained by a Fourier transformation of the sequence $\{k\sN(\tau)\hat{c}_X\sN(\tau):\tau=0,\cdots,N-1\}$, where $k\sN(\tau)$ is chosen suitably for $\tau=0,\cdots,N-1$.

The key feature in Definition~\ref{asympstat} is that in~(\ref{eq=asymptStat}) we average the time-varying autocovariance sequence $c_{\Y}(t,t+\tau)=\E\left\{\Y_{t}\Y_{t+\tau}\right\}$ across a time period $N$ to produce an average autocovariance across the time period, written as $\overline{c}_{\Y}\sN(\tau)$.
If this converges (in $N$) to a function of $\tau$, then by observing the modulated process over a suitably long time interval
we can recover the second-order properties of the stationary latent process.

We now wish to explore a more general assumption than that of asymptotic stationarity for modulated processes.
Specifically, we seek a larger class of models where consistent inference is still achievable. This will be smaller than the full class of models for $g_t$, as using a trivial example, if 
%Inference will be based on $\hat{c}_{\Y}\sN(\tau)$, or equivalently \textcolor{red}{its Fourier transform, the empirical Fourier spectrum or periodogram of $\Y_t$}. Even though we know the form of
%$\{g_t\}$, inference will not always be possible. As a trivial example,
$g_t\equiv 0$ always then we would not be able to infer properties of the generating mechanism of $X_t$.
%To be able to recover the parameters of the latent model, we \textcolor{red}{\sout{make the following assumption} define a new class of modulated processes}.
For consistent inference we propose the following class of modulated processes.
\begin{definition}[Modulated process with a significant correlation contribution]
\label{def=univariateSignificantCorrel}
%Assume that $\Y_t$ is specified by~\eqref{modeqn}. We say that $\Y_t$ is a \emph{modulated process with a significant correlation contribution} if for any $\tau\geq 0$ there exists two constants $N_\tau\geq \tau$ and $\alpha_\tau>0$ such that for
%$
%N\geq N_\tau,
%$
%\begin{equation}
%\label{eq=assumption1}
%\left|c_g\sN(\tau)\right| \geq \alpha_\tau.
%\end{equation}
Assume that $\Y_t$ is specified by~\eqref{modeqn}. We say that $\Y_t$ is a \emph{modulated process with a significant correlation contribution} if there exists a finite subset of non-negative lags $\Gamma\subset\N$ such that,
\begin{enumerate}
	\item The mapping $\btheta\mapsto \left\{c_X(\tau):\tau\in\Gamma\right\}$ is one-to-one (injective).
	\item For all lags $\tau\in\Gamma$, 
		\begin{equation}
		\label{eq=assumption1}
			\liminf\limits_{N\rightarrow\infty}\left|c_g\sN(\tau)\right| > 0,
		\end{equation}
		where $\liminf\limits_{N\rightarrow\infty}$ is the limit inferior.
\end{enumerate}
%We also define a \emph{modulated process with a significant lag-$L$ correlation contribution} when the above requirement is true for lags $0,\cdots, L$.
\end{definition}
Because of the symmetry of autocovariance sequences we do not need to consider $\tau<0$ in this definition. 
Point 1 of Definition~\ref{def=univariateSignificantCorrel} means that for any two distinct parameter vectors $\btheta,\btheta'\in\Theta$, there exists at least one lag $\tau$ in the finite set $\Gamma$ such that $c_X(\tau;\btheta)\neq c_X(\tau;\btheta')$. It is therefore an assumption about the latent process model.
The sequence $\left|c_g\sN(\tau)\right|$ is bounded above since the modulating sequence is assumed to be bounded above. Therefore the limit inferior in~\eqref{eq=assumption1} is always finite. 
We observe that, for $\tau\in\Gamma$,~\eqref{eq=assumption1} is equivalent to,
\begin{equation}
	\label{eq=equivalenceoflimitinf}
	\exists\alpha_\tau>0, \exists N_\tau\in\N, \forall N\in\N, N\geq N_\tau \Rightarrow \left|c_g\sN(\tau)\right| \geq \alpha_\tau,
\end{equation}
which we interpret as the fact that the sequence $\left|c_g\sN(\tau)\right|$ is bounded below for $N$ large enough.
For further understanding of Point 1 in Definition \ref{def=univariateSignificantCorrel} we provide the following two simple examples.
\begin{enumerate}
\item Let the latent process $\{X_t\}$ be an autoregressive process of order $p$, denoted AR($p$), with known mean zero and unknown innovation variance, and with the parameter set $\Theta$ that is a subset of $\R^{p+1}$. If the parameter set $\Theta$ is chosen appropriately, i.e. such that the roots of the characteristic equation all lie outside the unit circle, the Yule-Walker equations \citep{BrockwellDavisTimeSeries} show that   $\btheta\mapsto \left\{c_X(\tau; \btheta):\tau\in\Gamma\right\}$, where $\Gamma = \{0, \cdots,p\}$, is a one-to-one mapping. 
%This means that for two distincts parameter vectors $\btheta,\btheta'\in\Theta$, the finite sequences $ \left\{c_X(\tau; \btheta):\tau\in\Gamma\right\}$ and $\left\{c_X(\tau; \btheta'):\tau\in\Gamma\right\}$ are distinct, i.e there exists at least one lag $\tau\in\Gamma$ such that $c_X(\tau; \btheta')\neq c_X(\tau; \btheta)$.
Similarly if $\{X_t\}$ is a moving average process of order $q$, denoted MA($q$), with known mean an unknown innovation variance and if the parameter set $\Theta$ is chosen appropriately~\citep{dzhaparize83}, then the mapping $\btheta\mapsto \left\{c_X(\tau; \btheta):\tau\in\Gamma\right\}$, where $\Gamma = \{0, \cdots,q\}$, is one-to-one.
\item Let the latent process $\{X_t\}$ be the MA(2) process defined by,
\begin{equation}
	X_t = \sigma\left(\epsilon_t + \theta_2\epsilon_{t-2}\right), 
\end{equation}
where the innovations $\epsilon_t$ are i.i.d and have a standard normal distribution and $\sigma>0$. The parameters of the model are $(\theta_2, \sigma)$, and the parameter set $\Theta = \R\times\R\backslash\{0\}$ ensures that the mapping $\btheta\mapsto \left\{c_X(\tau; \btheta):\tau\in\Gamma\right\}$, where $\Gamma = \{0, 2\}$, is one-to-one. Note that observing lag-$1$ is not required here as we have assumed $\theta_1=0$ in the model. 
\end{enumerate}

The definition of a \emph{significant correlation contribution} constrains how much energy adds up for any fixed lag $\tau\in\Gamma$. We see directly from~\eqref{eq=expectedAutocovSequence1} that if we assume a significant correlation contribution, the expectation of the estimated autocovariance of $\Y_t$ does not vanish with the length of the observation $N$, at least for lags in $\Gamma$.
This allows for consistent estimation of the parameter $\btheta$ as we will see in Section \ref{sec=consistency}. As a trivial counterexample, assume for instance that $c_g\sN(\tau)$ goes to zero when $N$ goes to infinity. Then $\hat{c}_{\Y}\sN(\tau)$ in~\eqref{eq=biasedAutocovEst} goes to zero as well, independently of the parameter vector $\btheta$, resulting either in infeasible estimation or requiring a change of estimation approach.

Asymptotically stationary modulated processes are a subclass of modulated processes with a significant correlation contribution. 
Specifically, for the class of asymptotically stationary modulated processes,~\eqref{estimate} requires that $c_g\sN(\tau)$ converges to the non-zero quantity $R_g(\tau)$, which is a stronger requirement than~\eqref{eq=assumption1} where we only require an asymptotic positive lower bound rather than convergence.

\subsection{Missing observations}
\label{sec=missing}
%In the past the theory of modulated processes has been applied mostly to the problem of missing observations in time series.
%Missing observations in time series is an important topic as most time series estimation methods have been developed under the assumption of fully observed time series. For instance the Whittle likelihood is designed for fully observed samples. But missing points regularly occur in real-world data acquisition. This can be due to the sampling procedure or measurement issues for instance.
A particularly enticing use of modulated processes is to account for missing observations in stationary time series. 
Let $\{X_t:t\in\N\}$ be a stationary process. For each time point $t\in\N$, we set \citep{Parzen1963},
\begin{equation}
	g_t= \left\{
    \begin{array}{ll}
        0 & \mbox{if $X_t$ is missing} \\
        1 & \mbox{if $X_t$ is observed}
    \end{array}
\right.
.
\end{equation}
The process $\Y_t=g_t X_t$ is formed at all time points $t\in\N$, forming a modulated process in the sense of Definition~\ref{def=modulatedprocesses}.

An example where the missing observation pattern is deterministic and leads to an asymptotically stationary modulated process is the case of $(k,l)$-periodically missing data treated by \citet{R.H} and \citet{Parzen1963}. This corresponds to observing the $k$ first values, missing the $l$ next values, observing the $k$ next values, and so on. Note that \citet{Parzen1963} requires $k>l$ for non-parametric estimation of the spectral density of $X_t$ based on ~\eqref{estimate}.
Our model of modulated processes with significant correlation contribution allows for $k\leq l$, as long as we observe the lags in $\Gamma$.
A generalization of this missing data scheme was introduced by \citet{clinger1976} with an application to oceanography.

%Missing data can then be represented by the use of a modulating sequence made of zeros (observation is missed) and ones (observation is not missed). 
Missing observations can also occur according to a random mechanism. This can be modelled by a random modulation sequence taking values zero and one
\citep{Schein1965, Bloomfield1970}, when the random mechanism according to which missing points occur is independent from the observed process, which we shall assume.
Conditioning on the observed modulation function, we then return to the deterministic modulating sequence described in this paper. 
Most works, to our knowledge, have assumed some sort of stationarity for the random modulation sequence, i.e. that the sample autocovariance of the modulation sequence converges almost surely to a non-zero value at all lags~\citep{Dunsmuir1981b, Dunsmuir1981c}.
Some authors do not require such an assumption but have treated only specific models, usually autoregressive models \citep{Jones1980,Broersen2004}.
The definition of a modulated process with a significant correlation contribution in such a situation needs to be understood in a probabilistic fashion, i.e. we require that Property 2 of Definition \ref{def=univariateSignificantCorrel} be satisfied with probability one. Indeed, if one sees the general random experiment as a two-step experiment, where first the random modulating sequence $\{g_t\}$ is generated and observed and then a stationary process $\{X_t\}$ is modulated by this modulating sequence to produce $\{\Y_t\}$, then with probability one the moduating sequence $\{g_t\}$ in the first step makes $\{\Y_t\}$ a modulated process with significant correlation contribution. 
Such a situation will be described by saying that $\{\Y_t\}$ is a modulated process with an almost surely significant correlation contribution.
We shall now give a few examples of cases satisfying the stated conditions.
%We now provide examples of stationary processes with missing observations that can be modelled as modulated processes with \textcolor{red}{an almost surely} significant correlation contribution.
\begin{enumerate}
	\item Let $X_t$ be an AR($p$) Gaussian process with mean zero. If we set $\Gamma=\left\{ 0,\cdots,p\right\}$, and if the missing data occurs deterministically according to a $(k,l)$-periodic pattern, 
	%if $k\geq p$ then the resulting modulated process has a significant correlation contribution. This is because we are able to observe an infinite number of time the lags in $\Gamma$. We do not require any additional condition on $l$.
	$k\geq p$ is a sufficient condition for the resulting modulated process to have a significant correlation contribution. This is because we are able to observe an infinite number of time the lags in $\Gamma$. We do not require any additional condition on $l$.
	\item \label{ex=missingDataScheme} Let $X_t$ be an AR($p$) process, and consider the missing data scheme treated by \citet{Schein1965}, where the random mechanism is a sequence of Bernoulli i.i.d trials with identical probability of success (to be understood as \emph{observation} here) $0<p\leq 1$. 
	According to the strong central limit theorem, for all lag $\tau\in\N$,  $c_g\sN(\tau)$ converges a.s. to $p^2>0$ and therefore $\liminf\limits_{N\rightarrow\infty}\left|c_g\sN(\tau)\right| > 0$ a.s. Therefore the observed process is a modulated process with an almost surely significant correlation contribution.
	\item\label{example=missingDatasch} Consider the random mechanism where the sequence $\{g_t\}$ is generated according to
\begin{equation}
	g_t \sim \mathcal{B}(p_t),
\end{equation}
where $\mathcal{B}(p)$ represents the Bernoulli distribution with parameter $p$, and where we set 
\begin{equation}
	\label{eq:missing3}
	p_t = \mathcal{P} + A_p\cos\left(\omega_p t\right),
\end{equation}
with $0<\mathcal{P}<1$, $0\leq A_p<\min\left(\mathcal{P}, 1-\mathcal{P}\right)$ (which ensures $0<\mathcal{P}-A_p\leq p_t\leq1, \forall t\in\N$), and $\omega_p\in[-\pi,\pi]$.
The Bernoulli parameters $p_t$ as given by~\eqref{eq:missing3} will oscillate periodically around their mean value $\mathcal{P}$. This also leads to
$\liminf\limits_{N\rightarrow\infty}\left|c_g\sN(\tau)\right| > 0$ a.s., using the fact that $p_t$ is bounded below by $\mathcal{P}-A_p>0$.
\end{enumerate}
In section \ref{sec=missingSims} we will provide a simulation study based on example \ref{example=missingDatasch}. This is novel in comparison of previsouly studied missing observation schemes as we do not make an assumption of stationarity for the process $g_t$ \citep{Dunsmuir1981}.

\subsection{Sampling properties of modulated processes}
\label{sec=samplingproperties}
In this section we shall review and study some distributional properties of the periodogram of a modulated time series. \citet{Dunsmuir1981} used the periodogram as the basis for designing pseudo-likelihood 
methods for asymptotically stationary modulated time series, with an emphasis on treating the problem of missing data. 
Similarly, in Section \ref{sec=estimation} we will use the results of this section to formulate a pseudo-likelihood using the periodogram, for our class of modulated processes with significant correlation contribution. Herein we shall denote $\Omega_N$ the set of Fourier frequencies $\frac{2\pi}{N}\cdot\left( -\lceil\frac{N}{2}\rceil+1, \cdots, -1, 0, 1, \cdots, \lfloor\frac{N}{2}\rfloor \right)$.

We denote $\mathbf{\Y}=\{\Y_t:t=0,\cdots,N-1\}$ as a single realization of a length-$N$ sample of a modulated process $\{\Y_t\}$ defined in Definition \ref{def=modulatedprocesses}. The unobserved sample of the latent stationary process is denoted $\mathbf{X}=\{X_t:t=0,\cdots,N-1\}$ accordingly.
The squared modulus of the Fourier transform of the time series $\mathbf{X}$, known as the \emph{periodogram}, is a common statistic in stationary time series analysis~\citep{SpectralAnalysis}, and is given by
\begin{equation}
\label{eq=defPeriodogram}
	\hat{S}_X\sN(\omega) = \frac{1}{N}\left|\sum_{t=0}^{N-1}{X_t e^{-i\omega t}}\right|^2,\ \omega\in\R.
\end{equation}
Note that this quantity is $2\pi$-periodic, i.e. $\hat{S}_{X}\sN(\omega+2\pi)=\hat{S}_{X}\sN(\omega)$, $\omega\in\R$. The periodogram of the sample $\bold{X}$ is an asymptotically unbiased estimator of the spectral density of the stationary process $\{X_t\}$, i.e. $\allowbreak\lim_{N\rightarrow\infty}\E\{\hat{S}_X\sN(\omega);\btheta\}=2\pi S_X(\omega;\btheta)$ for all $\omega\in[-\pi,\pi)$~\citep{BrockwellDavisTimeSeries}. However the variance of the periodogram does not decrease to zero as the sample size increases.
A consistent nonparametric estimator of a smooth spectral density $S_X(\omega;\boldsymbol{\theta})$ of the latent process $\{X_t\}$, were it to be directly observed, could be obtained by smoothing the periodogram across frequencies~\citep[p. 235--253]{SpectralAnalysis}, as long as $S_X(\omega;\boldsymbol{\theta})$ is continuous.

	For the modulated process $\{\Y_t\}$, the latent time series $\{X_t\}$ is not observed, so we instead compute the periodogram of the modulated (and observed) process itself, $\hat{S}_{\Y}\sN(\omega)$, and we define the expected periodogram to be
	\[
		\overline{S}_{\Y}^{(N)}(\omega;\boldsymbol{\theta}) = \E\left\{\hat{S}_{\Y}\sN(\omega);\; \boldsymbol{\theta}\right\}, \ \ \omega\in\R.
	\]
Note that this quantity is also $2\pi$-periodic.
It is necessary to understand how modulation in the time domain will affect the expected periodogram.
Proposition \ref{prop=freqPropSbar} gives more insight on how $\overline{S}_{\Y}\sN(\omega;\boldsymbol{\theta})$ relates to the modulating sequence $g_t$ and the spectral density $S_X(\omega;\boldsymbol{\theta})$ of the latent stationary process $\{X_t\}$.

\begin{proposition}[Expectation of the periodogram of a modulated time series]
	\label{prop=freqPropSbar}
	The expectation of the periodogram of the modulated time series takes the form
	\begin{equation}
		\label{eq=Sbar1}
		\overline{S}_{\Y}\sN(\omega;\boldsymbol{\theta}) 
		= 2\pi\int_{-\pi}^\pi{S_X(\omega-\lambda;\boldsymbol{\theta})S_g\sN(\lambda)d\lambda}, \ \forall\omega\in\R,
	\end{equation}
	which is a periodic convolution.
	Here $S_g\sN(\lambda)$ is the squared value of the Fourier Transform 
	of the finite sequence $\{g_t\}_{t=0,\cdots,N-1}$ i.e.
	\[
		S_g\sN(\lambda) = \frac{1}{2\pi N}\left|\sum_{t=0}^{N-1}{g_t e^{-i\lambda t}}\right|^2,
	\]
	defined for $\lambda\in\R$ and which is $2\pi$ periodic.
\end{proposition}
\begin{proof}
	The proof for this proposition, which is a well-known result, can be found in~\citet[p.~562]{Dunsmuir1981}.
\end{proof}
When $g_t=1$ everywhere, which corresponds to observing the stationary latent process directly, the quantity $S_g\sN(\lambda)$ is the usual Féjer kernel \citep{Bloomfield2000} defined by,
\begin{equation}
\label{eq=FejerKernel}
	\mathcal{F}^{(N)}(\lambda) = \frac{\sin^2\left(\frac{N\lambda}{2}\right)}{2\pi N\sin^2\left(\frac{\lambda}{2}\right)}, \ \ \forall\lambda\in\mathbb{R}\setminus\Omega_N,
\end{equation}
which behaves asymptotically (as $N$ tends to infinity) as a Dirac delta-function centred at zero. This explains why the periodogram is, asymptotically, an unbiased estimator of the spectral density of a stationary process up to a multiplicative factor of $2\pi$~\citep{BrockwellDavisTimeSeries}.

When $g_t$ is such that the modulated process is asymptotically stationary, \citet{Dunsmuir1981b} approximate $\frac{1}{2\pi}\sum_{\tau=-\infty}^\infty{\gamma(\tau)e^{i\omega\tau}}$, where $\gamma(\tau) = R_g(\tau)c_X(\tau)$ using the notation of~\eqref{asymptotic}, for $\omega$ at Fourier frequencies by,
\begin{equation}
	\label{eq=DunsmuirSpectrum}
	\widetilde{S}_{\Y}^{(D)}(\omega;\btheta) = \frac{2\pi}{N}\sum_{\lambda\in\Omega_N}{S_X(\omega-\lambda;\btheta)S_g^{(N)}(\lambda)}.
\end{equation}

When $g_t$ is such that the modulated process $\Y_t$ has a significant correlation contribution, we derive the exact value of $\overline{S}_{\Y}\sN(\omega;\boldsymbol{\theta})$ by using the theoretical autocovariances of the latent model, in a similar fashion as in \citet{WhittleAdam} for stationary processes. This is the result of Proposition \ref{prop=expectationPeriodogram2}, which follows.
\begin{proposition}[Computation of the expected periodogram]
	\label{prop=expectationPeriodogram2}
	Let $\omega\in\R$.
	We have
	\begin{equation}
		\label{eq=computationOfPeriodogram}
		\overline{S}_{\Y}\sN(\omega;\boldsymbol{\theta}) = 
		2\mathcal{R}\left\{\sum_{\tau=0}^{N-1}{\overline{c}_{\Y}\sN(\tau;\boldsymbol{\theta})e^{-i\omega \tau}}\right\} 
		- \overline{c}_{\Y}\sN(0;\boldsymbol{\theta}),
	\end{equation}
	where $\overline c_{\Y}\sN(\tau;\bm\theta)$ is defined in~(\ref{eq=expectedAutocovSequence1}).
	By defining $\overline{c}_{\Y}\sN(-\tau;\boldsymbol{\theta}) = \overline{c}_{\Y}\sN(\tau;\boldsymbol{\theta})$ for $\tau=1, \cdots, N-1$ we can (equivalently) express this relationship as
	\[
		\overline{S}_{\Y}\sN(\omega;\boldsymbol{\theta}) =
		\sum_{\tau=-(N-1)}^{N-1}{\overline{c}_{\Y}\sN(\tau;\boldsymbol{\theta})e^{-i\omega \tau}},
	\]
\end{proposition}

\begin{proof}
The proof, which is standard~\citep[page 334]{BrockwellDavisTimeSeries}, follows directly from~\eqref{eq=defPeriodogram} and~\eqref{eq=expectedAutocovSequence1} in a few lines of algebra after aggregating along the diagonal of the covariance matrix.
\end{proof}
Therefore the expectation of the periodogram of $\bold{\Y}$ is the discrete Fourier transform of the expected sample autocovariance sequence. 
This is true even though we have not assumed stationarity; it is simply a consequence of the relation between the formal definitions of~\eqref{eq=expectedAutocovSequence1} and~\eqref{eq=defPeriodogram}.
Note that calculating the Fourier transform of the sequence $\overline{c}_{\Y}\sN(\tau;\theta)$ will always give a real-valued positive $\overline{S}_{\Y}\sN(\omega; \boldsymbol{\theta})$ for $\boldsymbol{\theta}\in\Theta$, as the latter is defined as the expectation of the squared modulus of the Fourier transform of the process. 

%This makes it always a suitable quantity to model the variance of the Fourier transform.

Proposition~\ref{prop=expectationPeriodogram2} can be used to compute the expected periodogram of an asymptotically stationary modulated process. In such cases, the difference between~\eqref{eq=DunsmuirSpectrum} and Proposition~\ref{prop=expectationPeriodogram2} is that~\eqref{eq=DunsmuirSpectrum} is a finite approximation of~\eqref{eq=Sbar1}, whereas Proposition~\ref{prop=expectationPeriodogram2} is exact. The difference occurs because~\eqref{eq=DunsmuirSpectrum} does not account for the bias of the periodogram that results from leakage (see \citet{WhittleAdam}), whereas these effects are naturally accounted for in Proposition~\ref{prop=expectationPeriodogram2}.
%If $\Y_t$ is an asymptotically stationary modulated process, we can use Proposition \ref{prop=expectationPeriodogram2}, since $\Y_t$ is also a modulated process with significant correlation contribution. The difference with using~\eqref{eq=DunsmuirSpectrum} boils down to the fact that it is a finite approximation of~\eqref{eq=Sbar1}. It does make a difference however, as for a stationary process, the convolution as described in~\eqref{eq=DunsmuirSpectrum} retrieves the spectral density, and does not account for the bias of the periodogram that results from leakage (see \citet{WhittleAdam}). There is no reason to expect those issues will vanish in the more general case of asymptotically stationary modulated processes.

To justify the use of the expected periodogram in the setting of modulated processes with a significant correlation contribution, 
we now consider what conditions are required for the expected periodogram to \emph{carry enough information} so that the parameter vector is identifiable within the parameter set $\Theta$. 

%\begin{proposition}[Information of the expected periodogram]
%\label{prop=identifiabilityViaPeriodogram}
	%Assume that the model family is such that for all $\btheta$ and $\widetilde{\btheta}$ non equal, the spectral densities $S_X\left(\omega;\btheta\right)$ and $S_X\left(\omega;\widetilde{\btheta}\right)$ are non equal on a Borel subset of $[-\pi,\pi]$ of positive measure. Also assume that the spectral densities are continuous functions in both parameters.
	%Then there exists a non-negative integer $L$ such that for all $\btheta$ and $\widetilde{\btheta}$ non equal, the families $\left\{c_X(\tau; \btheta);\tau=0,\cdots,L)\right\}$ and $\left\{c_X(\tau; \widetilde{\btheta});\tau=0,\cdots,L)\right\}$ are not equal.
	%Furthermore, if the modulated process has a $lag-$L$$ significant correlation contribution (i.e there exists $\alpha_\tau>0$ such that $\liminf\limits_{N\rightarrow\infty} \left|c_g\sN(\tau)\right| \geq\alpha_\tau$, for $\tau=0,\cdots,L$), the expected periodogram is a one-to-one mapping, for a sample size larger than $L$.
%\end{proposition}
%\begin{proof}
%See Appendix \ref{proof=identifiabilityViaPeriodogram}.
%\end{proof}
\begin{proposition}[Identifiability of the expected periodogram]
\label{prop=identifiabilityViaPeriodogram}
%Assume that there exists a non-negative integer $L$ such that for all $\btheta$ and $\widetilde{\btheta}$ distinct elements of $\Theta$, the families $\left\{c_X(\tau; \btheta);\tau=0,\cdots,L\right\}$ and $\left\{c_X(\tau; \widetilde{\btheta});\tau=0,\cdots,L\right\}$ are not equal.
If the modulated process has a significant correlation contribution, the expected periodogram is a one-to-one (i.e. injective) mapping from the parameter set $\Theta$ to the set of non-negative continuous functions on $[-\pi,\pi]$, for a large enough sample size. 
More specifically, for two distinct parameter vectors $\btheta$ and $\btheta'$, the expected periodograms $\overline{S}_{\Y}\sN(\omega;\btheta)$ and $\overline{S}_{\Y}\sN(\omega;\btheta')$ cannot be equal for all Fourier frequencies $\frac{2\pi}{N}\left( -\lceil\frac{N}{2}\rceil+1, \cdots, -1, 0, 1, \cdots, \lfloor\frac{N}{2}\rfloor \right)$.
\end{proposition}
\begin{proof}
		 Let $\btheta, \widetilde{\btheta}\in\Theta$ be distinct parameter vectors and let $N$ be a positive integer. Let $\Gamma$ be as given by Definition~\ref{def=univariateSignificantCorrel}.
%, denote $\alpha_\tau = \liminf\limits_{N\rightarrow\infty}\left|c_g\sN(\tau)\right|$. 
By the assumption of significant correlation contribution, the finite sequences $\left\{c_X(\tau; \btheta):\tau\in\Gamma\right\}$ and $\{c_X(\tau; \widetilde{\btheta}):\tau\in\Gamma\}$ are not equal. Since $\overline{c}_{\Y}\sN(\tau;\btheta) = c_g\sN(\tau)c_X(\tau;\btheta)$ for $\tau\in\Gamma$, and according to~\eqref{eq=equivalenceoflimitinf}, for $N$ large enough the sequences $\{\overline{c}\sN_{\Y}(\tau; \btheta):\tau\in\Gamma\}$ and $\{\overline{c}\sN_{\Y}(\tau; \widetilde{\btheta}):\tau\in\Gamma\}$ are not equal. Hence for $N$ large enough the sequences $\{\overline{c}\sN_{\Y}(\tau; \btheta):\tau=-(N-1),\cdots,N-1)\}$ and $\{\overline{c}\sN_{\Y}(\tau; \widetilde{\btheta}):\tau=-(N-1),\cdots,N-1)\}$ are not equal. Their finite Fourier transforms $\{\overline{S}_{\Y}\sN(\omega;\btheta):\omega\in\Omega_N\}$ and $\{\overline{S}_{\Y}\sN(\omega;\widetilde{\btheta}):\omega\in\Omega_N\}$, are by the bijective nature of the Fourier transform, not equal either.
\end{proof}
This means that for two distinct parameters vectors $\btheta,\btheta'\in\Theta$, we will have two distinct expected periodograms. This is a necessary condition for an estimation procedure based on the expected periodogram. 
We will propose such an estimation procedure in Section~\ref{sec=estimation}, and derive its consistency and convergence rate in Section~\ref{sec=consistency}.

\subsection{Bivariate modulated processes}
\label{sec=bivariate}
It is common in pratical applications to observe more than one time series at any time, and to analyse a set together. Often the series in the set are related via phase-shifts and other small temporal inhomogeneities, see e.g.~\citet{allen1996distinguishing,runstler2004modelling,allefeld2009mental,lilly2012analysis}.
Bivariate nonstationary processes can be challenging to model, as they may not be representable in the same nonstationary oscillatory family~\citep{tong1973some,tong1974time}.
To explore the nature of multivariate modulation, we shall investigate the representation 
of bivariate processes. For ease of exposition we shall represent such series using complex-valued time series, see~\citet{walker1993complex}. We shall continue to assume that the latent process, now denoted $Z_t$ for complex-valued processes, is Gaussian and zero-mean, leaving only the second order structure to be modelled.
For complex-valued processes both the autocovariance $c_Z(\tau)=\E\{Z_t^*Z_{t+\tau}\}$ (the star denotes conjugation) and the relation $r_Z(\tau)=\E\{Z_t Z_{t+\tau}\}$ sequences need to be modelled~\citep{walden2013rotary}. Complex-valued processes, unlike real-valued, no longer have a spectrum that needs to satisfy Hermitian symmetry, and if the series represents motion in the plane, the positive and negative frequencies represent clockwise and anti-clockwise rotations respectively.
Following the classical modelling framework~\citep{miller1969complex} for complex-valued processes we shall assume that the relation sequence takes the value zero for all lags. The complex-valued process is then said to be \emph{proper}. The assumption of propriety has the consequence of directly extending equation~\eqref{modeqn} to the complex-valued case from the real-valued case. Specifically, let $Z_t$ be a complex-valued Gaussian proper zero-mean process, a complex-valued modulated process is defined as one taking the form,
\begin{equation}
	\label{eq=complexMod}
	\widetilde{Z}_t = g_tZ_t,
\end{equation}
at all times $t\in\N$, where $g_t=\rho_t e^{i\phi_t}$ is a bounded modulation sequence.
We note that for complex-valued time series the modulation sequence is complex-valued. 
%Henceforth we shall favour the complex-valued representation~\eqref{eq=complexMod} for its simplicity, rather than the bivariate representation, because of the similar form~\eqref{eq=complexMod} has with that of univariate modulation, cf.~\eqref{modeqn}. 
With this definition, the modulation series $g_t$ accomplishes a time-dependent rescaling or expansion/dilation, from $\rho_t$, together with a time-dependent rotation, from $e^{i\phi_t}$.

%\textcolor{red}{
%This modulation allows for both amplitude and phase modulation of the complex-valued time series. One way to look at it is to consider that one observes the latent bivariate process via a measuring device that zooms in and out (this corresponds to the scaling factor) and rotates in either direction.
%}
The autocovariance of the complex-valued modulated process $\widetilde{Z}_t$ at times $t_1$ and $t_2$ is given by the conveniently simple form,
	\begin{equation}
		\label{autocovY1}
		\nonumber
		c_{\widetilde{Z}}(t_1,t_2;\boldsymbol{\theta})=\E\left\{\widetilde{Z}_{t_1}^*\widetilde{Z}_{t_2};\btheta\right\}=g_{t_1}^*g_{t_2}c_Z(t_2-t_1;\boldsymbol{\theta}) = \rho_{t_1}\rho_{t_2} e^{i(\phi_{t_2}-\phi_{t_1})}c_Z(t_2-t_1;\boldsymbol{\theta}),
	\end{equation}
	and $c_{\widetilde{Z}}(t_1,t_2;\boldsymbol{\theta})$ fully characterizes the process.
	Note that this quantity is not only a function of the lag $t_2-t_1$ as the process is no longer stationary.
	Similarly to the univariate case cf.~\eqref{eq=autocovOfg}, let $N$ be any positive integer, we define for $\tau=0,\cdots,N-1$,
	\begin{equation}
		\label{eq=cgcomplex}
		c_g\sN(\tau) = \frac{1}{N}\sum_{t=0}^{N-\tau-1}{g_t^*g_{t+\tau}}.
	\end{equation}
	Note that when $g_t$ is real-valued (\ref{eq=cgcomplex}) and (\ref{eq=autocovOfg}) are the same.
	We also extend the notion of a significant correlation contribution for complex-valued modulated processes, which naturally mimics Definition~\ref{def=univariateSignificantCorrel}.
We define the expected periodogram of a complex-valued modulated time series as
 $\overline{S}_{\widetilde{Z}}\sN(\omega) = \E\left\{
\hat{S}_{\widetilde{Z}}\sN(\omega) ; \btheta\right\}$, which can be computed efficiently similarly to Proposition~\ref{prop=expectationPeriodogram2} for the univariate case, by replacing $\Y_t$ by $\widetilde{Z}_t$ in~\eqref{eq=expectedAutocovSequence1} and~\eqref{eq=computationOfPeriodogram}.
	%\begin{equation}
		%\label{eq=expPeriod}
		%\overline{S}_{\widetilde{Z}}\sN(\omega;\boldsymbol{\theta}) = 
		%2\mathcal{R}\left\{\sum_{\tau=0}^{N-1}{\overline{c}_{\widetilde{Z}}(\tau;\boldsymbol{\theta})e^{-i\omega \tau}}\right\} 
		%- \overline{c}_{\widetilde{Z}}(0;\boldsymbol{\theta}),
	%\end{equation}
	%where
	%\begin{equation}
	%\label{CZtilde}
		%\overline{c}_{\widetilde{Z}}\sN(\tau;\boldsymbol{\theta}) = c_g\sN(\tau)c_Z(\tau;\boldsymbol{\theta}), \ \ \tau=0, \cdots, N-1.
	%\end{equation}
	%By defining $\overline{c}_{\Y}\sN(-\tau;\boldsymbol{\theta}) = \overline{c}_{\Y}\sN(\tau;\boldsymbol{\theta})^*$ for $\tau=1, \cdots, N-1$ we have the following equivalent relation,
	%\[
		%\overline{S}_{\widetilde{Z}}\sN(\omega;\boldsymbol{\theta}) =
		%\sum_{\tau=-(N-1)}^{N-1}{\overline{c}_{\widetilde{Z}}\sN(\tau;\boldsymbol{\theta})e^{-i\omega \tau}}.
	%\]
%Similarly, Proposition \ref{prop=identifiabilityViaPeriodogram} can readily be extended to the complex-valued process $\widetilde{Z}_t$.

A univariate real-valued modulated process is stationary if and only if the modulating sequence is a constant.
A necessary and sufficient condition on the modulating sequence for the complex-valued modulated process~\eqref{eq=complexMod} to be stationary is more complicated to obtain, and is determined in the following proposition.

\begin{proposition}[Stationary bivariate modulated processes]
	\label{prop=stationaryModulatedProcesses}
	Let $\widetilde{Z}_t$ be the complex-valued modulated process defined in~\eqref{eq=complexMod}.
	First, assume the latent process $\{Z_t\}$ is a white noise process.
	Then the modulated process $\{\widetilde{Z}_t\}$ is stationary if and only if the modulating sequence $g_t = \rho_te^{i\phi_t}$ is 
	of constant modulus, i.e. $\rho_t=a\geq0$. In such case the modulated process is a white noise process with variance $a^2 \E\{|Z_0|^2\}$.

	More generally, assume the stationary latent process $\{Z_t\}$ is not a white noise process, and let 
$\mu = \gcd\{\tau\neq0\in\N :|c_Z(\tau;\boldsymbol{\theta})|>0\}$ where $\gcd$ denotes the greatest common divisor.
	Then the modulated process is stationary if and only if $\{g_t\}$ is zero everywhere or if there
	exists two constants $a>0$ and $\gamma\in[-\pi,\pi)$ such that for all $t\in\N$, letting $r=t\bmod\mu$ be the remainder of $t$ divided by $\mu$,
	\begin{eqnarray*}
		\rho_t &=& a\\
		\phi_t &=& \phi_{r} + \gamma \left\lfloor{\frac{t}{\mu}}\right\rfloor \mod 2\pi,
	\end{eqnarray*}
	where $\left\lfloor{\frac{t}{\mu}}\right\rfloor$
	denotes the floor of $\frac{t}{\mu}$ and$\mod 2\pi$ indicates that the equality is true up to an additive multiple of $2\pi$.
	In this case the spectral density of the modulated process $\{\widetilde{Z}_t\}$ is
	\begin{equation*}
		S_{\widetilde{Z}}(\omega) = a^2S_Z\left(\omega-\frac{\gamma}{\mu}\right).
	\end{equation*}
\end{proposition}

\begin{proof}
	See appendix \ref{proof=stationaryModulatedProcesses}.
\end{proof}
The value of $\mu$ in Proposition \ref{prop=stationaryModulatedProcesses} depends on the location of zeros in the covariance sequence of the latent process.
In particular, if $|c_Z(1;\boldsymbol{\theta})|>0$ then $\mu = 1$ and $\widetilde{Z}_t$ is stationary only if there exists a constant $\gamma\in\R$ such that for all $t\in\N$, $\phi_t = \phi_0 + \gamma t \mod 2\pi$. If $|c_Z(2;\boldsymbol{\theta})|>0$ but $|c_Z(\tau;\boldsymbol{\theta})|=0$ for all $\tau\in\N, \tau\neq0,2$, then $\mu=2$ (this can occur with a second-order moving average process for instance). In that case the modulated process $\widetilde{Z}_t$ is stationary if and only if there exists a constant $\gamma\in[-\pi,\pi)$ such that for all $t\in\N$, $\phi_t = \phi_0 + \gamma \frac{t}{2} \mod 2\pi$ if $t$ is even, or $\phi_t = \phi_1 + \gamma \frac{t-1}{2} \mod 2\pi$ if $t$ is odd.

%In the case of a stationary bivariate modulated process, again one may propose to obtain an estimator of $\btheta$ by recovering the latent process according to $Z_t = g_t^{-1} \widetilde{Z}_t$ when $g_t$ is known. However for many physical examples, we do not simply observe a given process but instead an aggregation of many different processes $Z_t^{(1)}+Z_t^{(2)}+...+Z_t^{(P)}$. These would often be of a different archetypal time form -- see for example the class of unobserved components model \citep{harvey1990forecasting}. In such instances the problem cannot be inverted in this way. This is indeed the case for our real-world data example, where for each bivariate time series that is observed, the model will be composed of a Matérn process and a modulated complex-valued Ornstein-Ulhenbeck process.

\subsubsection{A time-varying bivariate autoregressive process}
\label{sec=tvAR}
We now introduce the specific bivariate autoregressive model that will be used in our real-world data application. 
We consider the discrete-time complex-valued process $\{\widetilde{Z}_t:t\in\N\}$, defined by
\begin{eqnarray}
	\label{eq=tvARCdef}
	\widetilde{Z}_t &=& r e^{i\beta_t}\widetilde{Z}_{t-1} + \epsilon_t, \ \ \ \ \ \ \ t\geq1,\ \ 0\leq r < 1, \ \ \beta_t\in\R,\\
	\nonumber\widetilde{Z}_0 &\sim& \mathcal{N}_{C}\left(0, \frac{\sigma^2}{1-r^2}\right),\ \ \ \ \sigma>0,\\
	\nonumber\epsilon_t &\sim& \mathcal{N}_{C}\left(0, \sigma^2\right),
\end{eqnarray}
where $\mathcal{N}_{C}\left(0, \sigma^2\right)$ denotes the complex-valued normal distribution with mean $0$ and variance $\sigma^2$, and with i.i.d real and imaginary part. Note that the real and imaginary parts of $\epsilon_t$ then each have variance $\sigma^2/2$.
Here $0\leq r < 1$ is commonly known as either the autoregressive or the damping parameter, ensuring the mean-reversion of the process.
By mean-reversion we mean that, beginning at any time $t$, we have $\lim_{\tau\to\infty} \E\left\{\widetilde{Z}_{t+\tau}|\widetilde{Z}_t\right\} = 0$, i.e. irrespective of the size of the perturbation $\epsilon_t$ at time $t$, the process is expected to return to its mean of 0. This is seen from the following inductive relationship,
\begin{equation*}
	\tilde{Z}_{t+\tau} = r^\tau e^{i\sum_{j=1}^{\tau}{\beta_{t+j}}}\tilde{Z}_t + \sum_{j=1}^\tau {r^{\tau-j}e^{i\sum_{k=j+1}^\tau\beta_k}\epsilon_{t+j}},  \ \ \tau\geq 0,
\end{equation*}
which leads to
\begin{equation*}
	\E\left\{\tilde{Z}_{t+\tau}|\tilde{Z}_t\right\} = r^\tau e^{i\sum_{j=1}^\tau{\beta_{t+j}}}\tilde{Z}_t,
\end{equation*}
which goes to zero exponentially as $\tau$ goes to infinity, since $0\leq r<1$.
A damping parameter $r$ close to 1 will lead to a slowly-decaying autocorrelation sequence.
A value of $r$ close to 0 will lead to a process with very short memory, with the limiting behaviour of a white noise process as $r\rightarrow0$.
The parameter $\beta_t$ is a known, dimensionless time-varying frequency, which we shall take within the interval $[-\pi,\pi)$ without loss of generality.

The process \eqref{eq=tvARCdef} is a nonstationary version of the complex-valued first order autoregressive process \citep{AdamWidelyLinear} introduced by~\citet{le1988note}, and also a discrete-time analogue of the complex-valued Ornstein-Ulhenbeck (OU) process \citep{Arato1962estimation} with time-varying oscillation frequency. We now prove in Proposition \ref{prop=tvarmodulated} that the model defined in (\ref{eq=tvARCdef}) belongs to our class of bivariate modulated processes.

\begin{proposition}[Modulated process representation]
	\label{prop=tvarmodulated}
	Let  $\{\widetilde{Z}_t\}$ be the process defined in (\ref{eq=tvARCdef}).
	There exists a unit-magnitude complex-valued modulating sequence $g_t$, 
	and a stationary complex-valued proper process $\{Z_t\}$ such that $\{\widetilde{Z}_t\}$ is the modulation of $\{Z_t\}$ by the non-random sequence $\{g_t\}$.
	More explicitly, we have $\widetilde{Z}_t = g_t Z_t$, for all $t\in\N$, where,
	\begin{eqnarray}
		\label{eq=tvARcModulatingSeq}
		g_t &=& e^{i\sum_{u=1}^{t}{\beta_u}},\\ 
		\nonumber Z_t &=& r Z_{t-1} + \epsilon_t', \ \ t\geq 1,
	\end{eqnarray}
	and $Z_0\sim\mathcal{N}_{C}(0, \sigma^2/(1-r^2))$. The process $\epsilon_t'$ is a Gaussian white noise process with the same properties
	(zero-mean, variance $\sigma^2$ and independence of real and imaginary parts) as those of $\epsilon_t$. Defined as such, the latent complex-valued process $Z_t$ is stationary and proper.
\end{proposition}

\begin{proof}
	See Appendix \ref{proof=tvarmodulated}.
\end{proof}
The stationary latent process $Z_t$ defined in~\eqref{eq=tvARcModulatingSeq} is a stationary complex-valued first order autoregressive process, and is Gaussian. Its autocovariance sequence is given by,
\begin{equation*}
	c_Z(\tau) = \frac{\sigma^2}{1-r^2}r^\tau, \ \ \tau\in\Z.
\end{equation*}
It is easy to verify that the mapping $\left(r, \sigma\right)\mapsto \left(c_Z(0), c_Z(1)\right)$ is a one-to-one mapping. 
In the following proposition, we stipulate a sufficient condition on the frequencies $\beta_t$ so that the process defined in $(\ref{eq=tvARCdef})$ satisfies our assumption of significant correlation contribution, when represented as a modulated process as defined in Proposition \ref{prop=tvarmodulated}.
\begin{proposition}
	\label{prop=tvARcCorrelationContribution}
	Let $\tilde{Z_t}$ be the process defined by (\ref{eq=tvARCdef}).  Assume that there exists $\Xi\in[-\pi,\pi)$ and $0\leq\Delta\leq\frac{\pi}{2}$ such that for all $t\in\N$, $\left|\beta_t-\Xi\right|\leq \Delta$. Then 
	%there exists an integer $L$, greater than or equal to 1, such that 
	$\tilde{Z_t}$ is a modulated process with significant correlation contribution.
\end{proposition}
\begin{proof}
	See Appendix \ref{proof=tvARcCorrelationContribution}.
\end{proof}

Hence the complex-valued autoregressive process defined by (\ref{eq=tvARCdef}) belongs to the class of processes with a significant correlation contribution, and the expected periodogram is a one-to-one mapping from the parameter set $[0,1)\times[0,\infty)$ to the set of non-negative continuous functions on $[-\pi,\pi]$, according to Proposition~\ref{prop=identifiabilityViaPeriodogram}---a proposition which is readily extended to the complex-valued processes of this section.
  
\section{Parametric estimation of modulated processes}
\label{sec=estimation}
We have explored a class of univariate and bivariate modulated processes. The next stage is to describe their efficient inference.
In this section we describe how the parameters of the latent model for $\{X_t\}$ can be inferred from observing a single realization of the modulated process $\{\Y_t\}$.
Most authors have focused on the problem of estimating modulated processes under the assumption of asymptotic stationarity as defined in Definition \ref{asympstat} \citep{Parzen1963,  Dunsmuir1981c, Dunsmuir1981, Iacobucci2003}. Although non-parametric estimates have been the key concern in most of the relevant literature, there have been instances where parametric estimation has been considered, see for instance \citet{Dunsmuir1981c}. Parametric estimation ensures that the estimated autocovariance sequence $\hat{c}_X(\tau)$ is non-negative definite, as opposed to using non-parametric estimates of the form given in~\eqref{estimate}. Parametric estimation is also preferable when the true model is known, as it uses the observed degrees of freedom more efficiently. Herein we consider the problem of parametric estimation for our class of modulated processes with a significant correlation contribution, which, we recall, is a generalization of asymptotically stationary modulated processes. We propose an adaptation of the Whittle likelihood \citep{whittle1953estimation}, based on the expected periodogram. 

We wish to infer the parameter vector $\boldsymbol{\theta}$ of the latent univariate stationary process $\{X_t\}$ within the parameter set $\Theta$, based on the sample $\bold{\Y}=\Y_0,\cdots,\Y_{N-1}$ and the known modulating sequence $\{g_t:t=0,\cdots,N-1\}$. 
Because it has been assumed that the latent process is a zero-mean Gaussian process, the same is true for the modulated process. The vector $\mathbf{\Y}$ is multivariate Gaussian with an expected $N\times N$ autocovariance matrix
$C_{\Y}(\boldsymbol{\theta}) = \left\{c_{\Y}(t_1,t_2;{\boldsymbol{\theta}})\right\}$ for $t_1,t_2=0,\cdots,N-1$, where the components of this matrix are given by $c_{\Y}(t_1,t_2;{\boldsymbol{\theta}}) = g_{t_1} g_{t_2} c_X(t_2-t_1;{\boldsymbol{\theta}})$.
However, the parameter vector $\boldsymbol{\theta}$ of the latent process $\{X_t\}$ can be uniquely determined from the modulated process $\{\Y_t\}$ only if $\boldsymbol{\theta}\rightarrow \left\{c_{\Y}(t_1,t_2;{\boldsymbol{\theta}}):t_1,t_2\in\N\right\}$ is injective, i.e. there is no ${\boldsymbol{\theta}}'\in\Theta$ such that $\btheta\neq\btheta'$ and $c_{\Y}(t_1,t_2;{\boldsymbol{\theta}}) = c_{\Y}(t_1,t_2;{\boldsymbol{\theta'}})\ \forall t_1,t_2\in\N$.
This condition is clearly achieved under the assumption of a modulated process with significant correlation contribution.
The negative of the exact time-domain Gaussian log-likelihood is proportional to

		\begin{equation}
			\label{eq=timedomainLKH}
			\ell_{G}(\boldsymbol{\theta}) = 
			 \frac{1}{N}\log{\left|C_{\Y}(\boldsymbol{\theta})\right|} 
			+ \frac{1}{N}\mathbf{\Y}^T C_{\Y}(\boldsymbol{\theta})^{-1}\mathbf{\Y},
		\end{equation}
		where $\left|C_{\Y}(\boldsymbol{\theta})\right|$ denotes the determinant of $C_{\Y}(\boldsymbol{\theta})$. Note that one may need to remove from $\bold{\Y}$ points where $g_t$ is zero, to ensure that the determinant of the covariance matrix is non-zero, and since those observations carry no information about $\btheta$. We minimize $\ell_t$ to obtain the time-domain maximum likelihood estimator (MLE), i.e.
		\[
			\boldsymbol{\hat{\theta}}_{G} = 
			\arg \min_{\btheta\in\Theta}\ell_t(\boldsymbol{\theta}).
		\]
%where additive constants have been excluded from $\ell_t$ in~\eqref{eq=timedomainLKH} as they do not affect the parameter choice made by maximizing the likelihood.
Parameter estimation based on time-domain likelihood has several drawbacks in the context of modulated processes.
For a large sample size $N$, computing the determinant of the covariance matrix is expensive, requiring $\mathcal{O}(N^3)$ elementary operations in general (although in specific cases such as for Markovian processes the likelihood is obtained in only $\mathcal{O}(N)$ computations).
Moreover each computation of the parametric covariance matrix $C_{\Y}(\boldsymbol{\theta})$ within the exact likelihood requires $\mathcal{O}(N^2)$ operations, compared to $\mathcal{O}(N)$ operations in the case of a stationary process.

We propose a computationally efficient estimation method for the parameters of the latent model based on the periodogram of the modulated time series. First recall that for the stationary time series $\{X_t\}$, making use of the Toeplitz property of the autocovariance matrix, one can approximate the log-likelihood using the Whittle likelihood \citep{whittle1953estimation}, which once discretized is evaluated by
\begin{equation}
	\label{eq=whittleLKH}
	\ell_W(\btheta) = \frac{1}{N}\sum_{\omega\in\Omega_N}\left\{\log S_X(\omega;\btheta) + \frac{\hat{S}_X^{(N)}(\omega)}{S_X(\omega;\btheta)}\right\},
\end{equation}
where again $\Omega_N$ is the set of Fourier frequencies $\frac{2\pi}{N}\cdot\left( -\lceil\frac{N}{2}\rceil+1, \cdots, -1, 0, 1, \cdots, \lfloor\frac{N}{2}\rfloor \right)$. This pseudo-likelihood has the benefit of $\mathcal{O}(N\log N)$ computational complexity using the Discrete Fourier Transform, with the resulting maximum likelihood estimator being asymptotically equivalent to the time-domain log-likelihood. 
We adapt this pseudo-likelihood procedure to modulated processes with significant correlation contribution in the following definition.
\begin{definition}[Spectral maximum pseudo-likelihood estimator for univariate modulated processes]
	\label{def=ourEstimator}
	Let $\{\Y_t\}$ be a modulated process with significant correlation contribution and let $\bf \Y$ be its length-$N$ sample. We define the following pseudo-likelihood:
	\begin{equation}
		\label{eq=newLKH}
		\ell_M(\boldsymbol{\theta}) = 
		 \frac{1}{N}\sum_{\omega\in\Omega_N}\left\{
		\log{\overline{S}\sN_{\Y}(\omega;\boldsymbol{\theta})} 
		+ \frac{\hat{S}_{\Y}\sN(\omega)}{\overline{S}\sN_{\Y}(\omega;\boldsymbol{\theta})}
		\right\},
	\end{equation}
	where $\overline{S}^{(N)}_{\Y}(\omega)$ is defined in Section \ref{sec=samplingproperties} as the expectation of the periodogram of the modulated time series, and is computed using Proposition \ref{prop=expectationPeriodogram2}.
	The corresponding estimator of the parameter vector $\boldsymbol{\theta}$ is obtained by a minimization procedure over the parameter set,
	\[
		\hat{\boldsymbol{\theta}}_M\sN = 
		\arg\min_{\btheta\in\Theta}\ell_M(\boldsymbol{\theta}).
	\]

\end{definition}

The sequence $\{c_g\sN(\tau): \tau=0,\cdots,N-1\}$ defined in (\ref{eq=autocovOfg}) requires $\mathcal{O}(N^2)$ computations in the most general case. This initial step is carried out independently of inferring the parameter of interest $\btheta$. Then any computation of $\{\overline{S}\sN_{\Y}(\omega;\boldsymbol{\theta}):\omega\in\Omega_N\}$ for any value of the parameter vector $\btheta$ will require $\mathcal{O}(N\log N)$ computations, since we can compute $\{\overline{c}_{\Y}\sN(\tau;\btheta):\tau=0, \cdots, N-1\}$ in $\mathcal{O}(N)$ computations using~\eqref{eq=expectedAutocovSequence1} and the precomputed $\{c_g\sN(\tau): \tau=0,\cdots,N-1\}$, and 
the quantity $\{\overline{S}\sN_{\Y}(\omega;\boldsymbol{\theta}):\omega\in\Omega_N\}$ is then computed via a fast Fourier transform.
The reason for separating this initial $\mathcal{O}(N^2)$ step from the rest of the computation is that it is carried out independently of the parameter value, and therefore outside any call to a minimization procedure over the parameter set $\Theta$ involving the expected periodogram. 

In the trivial case of a modulation sequence equal to 1 everywhere, then the likelihood of Definition~\ref{def=ourEstimator} does not exactly equal the Whittle likelihood of~\eqref{eq=whittleLKH}. This is because the spectral density $S_X(\omega)$ would be replaced by the expected periodogram $\overline{S}^{(N)}_{\Y}(\omega)$, which is the convolution of the true spectral density with the Fej\'er kernel (see~\eqref{eq=DunsmuirSpectrum}). For stationary time series, this type of estimator was investigated in~\cite{WhittleAdam}, and was found to significantly reduce bias and error in parameter estimation as compared with standard Whittle estimation.
For modulated processes that are asymptotically stationary, this signifies the difference between using~\eqref{eq=DunsmuirSpectrum} and the quantity defined by Proposition~\ref{prop=expectationPeriodogram2} to fit the periodogram.

The same estimator to~\eqref{def=ourEstimator} can be used for the complex-valued time series $\widetilde{Z}_t$ considered in Section \ref{sec=bivariate}, i.e. we define our estimator,
\begin{equation}
	\label{eq=complexEstimator}
		\hat{\boldsymbol{\theta}}_M\sN = 
		\arg\min_{\btheta\in\Theta}\ell_M(\boldsymbol{\theta}),
\end{equation}
with the objective function given by,
\begin{equation}
	\label{eq=newLKHcomplex}
	\ell_M(\boldsymbol{\theta}) = 
		\frac{1}{N}\sum_{\omega\in\Omega_N}
		{
		\left\{
		\log{\overline{S}\sN_{\widetilde{Z}}(\omega;\boldsymbol{\theta})} 
		+ \frac{\hat{S}_{\widetilde{Z}}\sN(\omega)}{\overline{S}\sN_{\widetilde{Z}}(\omega;\boldsymbol{\theta})}
		\right\}
		},
\end{equation}
The comments on computational aspects hold for the complex-valued case as well.
In Section~\ref{sec=consistency}, we will prove consistency of the frequency-domain estimator~\eqref{eq=newLKH} and its $\mathcal{O}(N^{-1/2})$ convergence rate.

\section{Applications}
\subsection{Application to Oceanographic Data}
\label{sec=realdata}
In this section we analyse real-world data from the Global Drifter Program (GDP). Specifically we model jointly the latitudinal and longitudinal velocities obtained from instruments known as drifters, which freely drift according to ocean surface flows~\citep{sykulski2016Lagrangian}.
Those velocities are modelled as the aggregation of two independent complex-valued processes, one of which is nonstationary and which we model as the complex-valued AR(1) process described in Section~\ref{sec=tvAR}. We use our estimator~\eqref{eq=complexEstimator} to infer physical quantities that describe the ocean surface currents.
To scrutinize our results and compare with alternative approaches, we also present two simulation studies, the first one being based on a dynamical model of the ocean surface currents and the second one being a simulated version of the model of Section~\ref{sec=tvAR}. 
All data and code used in this paper is available for download at \texttt{http://www.ucl.ac.uk/statistics/research/spg/software} and all results in this section and Section~\ref{sec=missingSims} are exactly reproducible.
\subsubsection{The Global Drifter Program}
\label{sec=GDPapp}
The GDP database (www.aoml.noaa.gov/phod/dac) is a collection of measurements obtained from buoys known as surface drifters, which drift freely with ocean currents and regularly communicate measurements to passing satellites at unequally spaced  time intervals averaging 1.4hrs. The data is then interpolated onto a regular temporal grid using the approach of~\citet{elipot2016global}. The measurements include position, and often sea surface and temperature. In total, over 11,000 drifters have been deployed, with approximately 70 million position recordings obtained. The analysis of this data is crucial to our understanding of ocean circulation \citep{lumpkin07}, which is known to play a primary role in determining the global climate system, see e.g.~\citet{andrews2012forcing}. Furthermore, GDP data is used to understand the dispersion characteristics of the ocean, which are critical in correctly modelling oil spills \citep{abascal2010analysis} and more generally assist in developing theoretical understanding of ocean fluid dynamics \citep{griffa2007lagrangian}, which is necessary for global climate modelling. 

In Fig.~\ref{Equatorial}(a), we display in the left panel the trajectories of 200 drifters which either traverse or are near the equator, interpolated for this application onto a 2 hour grid from raw position fixes available at the GDP web site. We focus on a single drifter trajectory, drifter ID\#43594, in panels (b) and (c), displaying both its latitudinal position and velocity respectively, the latter of which is obtained by differencing the positions. This velocity time series is nonstationary, as it has oscillations which appear to be modulated and change in frequency over time. The oscillations are known as {\em inertial oscillations}---one of the most ubiquitous and readily observable features of the ocean currents accounting for approximately half of the kinetic energy in the upper ocean \citep{ferrari2009arfm}. Inertial oscillations arise due to the deviation of the rotating earth from a purely spherical geometry, together with the appearance of the Coriolis force in the rotating reference frame of an earth-based observer \citep{early2012forces}. The modulation of these oscillations occurs because the drifters are changing latitude---and the {\em Coriolis frequency}, denoted $f$, is equal to twice the rotation rate of the Earth $\Omega$, multiplied by the sine of the latitude $\zeta$, i.e. $f=2\Omega \sin \zeta$ radians per second. The rotation rate of the Earth $\Omega$ is computed as $2\pi/T$ where $T$ is one sidereal day in seconds. Note that the Coriolis frequency $f$ is a signed quantity, implying that oscillations occur in opposite rotational clockwise/anti-clockwise directions from one hemisphere to the other. The Coriolis frequency is positive in the Northern hemisphere whereas the oscillations occur in the mathematically negative sense. 
Therefore we define the inertial frequency $\inerf = -f/2\pi K$ as the negative of the Coriolis frequency divided by $2\pi K$, where $K$ is one solar day in seconds, so that $\inerf$ is in cycles per day. The entire drifter dataset is split into segments of 60 inertial periods in length, accounting for the variation of the inertial period along drifter trajectories, and with 50\% overlap between segments.  The standard deviation of the inertial frequency along each data segment is taken, and the 200 segments exhibiting the largest ratio of the standard deviation of the inertial frequency, to the magnitude of its mean value along the segment, are identified for use in this study.  These exhibit the largest fractional changes in the inertial frequency, and as shown in Fig.~\ref{Equatorial}(a), are located in the vicinity of the equatorial region where inertial frequency vanishes.

\begin{figure}[h!]
\centering
\includegraphics[width=0.9\textwidth]{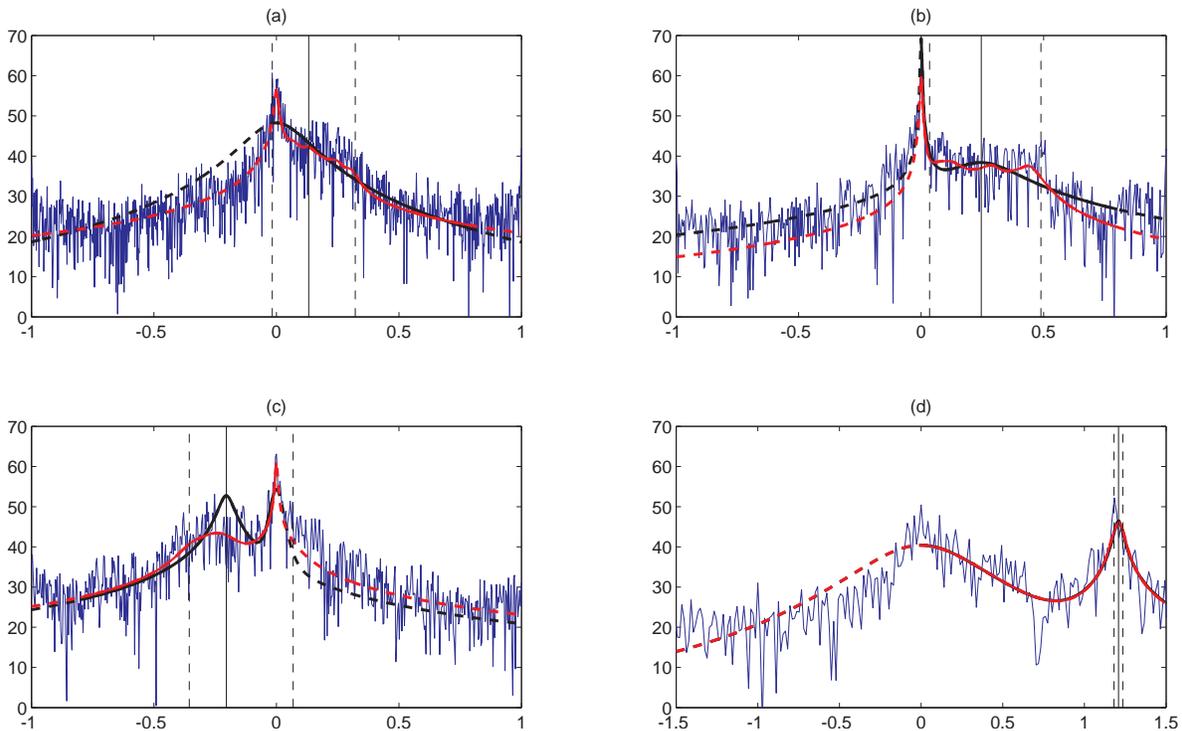}
\caption{
Fitted variance of the discrete Fourier transform using either the stationary model (in black) or the nonstationary model (in red) to the periodogram (in blue) for segments of data from drifter IDs
%(a) \#92629, (b) \#81896, (c) \#71845, and (d) \#44312. 
(a) \#79243, (b) \#54656, (c) \#71845 and (d) \#44312.
The solid black vertical line is the average inertial frequency, and the dashed vertical black lines are the minimum and maximum observed inertial frequency over the observed time window. The models are fit in the frequency range of 0 to 0.8 cycles per day in (a)--(c), and from 0 to 1.5 cycles per day in (d) as this drifter is at a higher latitude of $37^\circ$ S where inertial oscillations occur at a frequency of about 1.2 cycles per day. The fitted models are shown in solid lines within the frequency range, and in dashed lines outside the frequency range.}
\label{DrifterAnalysis} 
\end{figure}
%\begin{figure}[h!]
%\centering
%\includegraphics[width=0.9\textwidth]{qqplot1.eps}
%\caption{\label{DrifterAnalysisQQplot} QQplot of the ratio periodogram over spectral density verses the exponential distribution for drifter IDs (a) \#92629, (b) \#81896, (c) \#71845, and (d) \#44312. The data is limited to the frequency ranges that were used for the fitting procedure.}
%\end{figure}

\subsubsection{Stochastic modelling}
The stochastic modelling of Lagrangian trajectories was investigated in \citet{sykulski2016Lagrangian}, where the term ``Lagrangian" is used because the moving object making the observations (i.e. the drifter) is the frame of reference, as opposed to fixed-point measurements known as {\em Eulerian} observations. In that paper, the Lagrangian velocity time series was modelled as a complex-valued time series, with the following 6-parameter power spectral density:
\begin{eqnarray}
\label{driftermodel}
S(\omega) &=& \frac{A^2}{(\omega-\inerf)^2+\lambda^2}+\frac{B^2}{\left(\omega^2+h^2\right)^\alpha},\\
\nonumber &&A>0, \;\;\; \lambda>0, \;\;\; \inerf\in[-\pi,\pi], \;\;\; h>0, \;\;\; B>0, \;\;\; \alpha>\frac{1}{2},
\end{eqnarray}
where $\omega$ is given in cycles per day. The first component of~\eqref{driftermodel} is the spectral density of a complex Ornstein-Uhlenbeck (OU) process \citep{Arato1962estimation}, and is used to describe the effect of inertial oscillations at frequency $\inerf$. 
Denoting $\widetilde{Z}_{\text{OU}}(t)$ the OU component, where $\widetilde{Z}_{\text{OU}}(t)$ is complex-valued, these oscillations are described by the following stochastic differential equation (SDE),
\begin{equation}
	\label{eq=OU_SDE_stationary}
	d\widetilde{Z}_{\text{OU}}(t) = (-\lambda + i2\pi\inerf)\widetilde{Z}_{\text{OU}}(t)dt + AdW(t),
\end{equation}
where $t$ is expressed in days and $W(t)$ is a complex-valued Brownian process with independent real and imaginary parts. The damping parameter $\lambda>0$ ensures that the OU process is mean-reverting.
The corresponding continuous complex-valued autocovariance is given by
\begin{equation*}
	s(\tau) = \frac{A^2}{2\lambda}\exp\left\{-\lambda|\tau|+i2\pi\omega\tau\right\},
\end{equation*}
and the sampled process \citep{ARATO19991} $\widetilde{Z}_{\text{OU},t} = \widetilde{Z}_{\text{OU}}(t\Delta)$, where $\Delta=1/12$ day is the sampling rate corresponding to the 2hr grid, is a complex-valued AR(1),
\begin{equation}
	\label{eq=sampled_stationary_OU}
	\widetilde{Z}_{\text{OU},t} = re^{i2\pi\inerf\Delta}\widetilde{Z}_{\text{OU}, {t-1}} + \epsilon_t.
\end{equation}
Here $\epsilon_t$ is a Gaussian complex-valued white noise process with independent real and imaginary parts and variance $\sigma^2$. The autocovariance sequence of the stationary sampled process is given by,
\begin{equation}
	c_{\widetilde{Z}_{\text{OU}}} = \frac{\sigma^2}{1-r^2}r^{\tau\Delta}.
\end{equation}
The transformation between the parameters of the complex-valued OU and the complex-valued AR(1) are given by,
\begin{equation}
\label{eq=transform_params_OU_AR}
        \sigma^2 = \frac{A^2(1-e^{-\lambda\Delta})}{2\lambda\Delta}, \ \ \ \
        r = e^{-\lambda\Delta}.
\end{equation}

The second component of~\eqref{driftermodel} is the spectral density of a stationary proper Mat\'ern process \citep{gneiting2010matern}, denoted $Z_{M,t}$, and is used to describe two-dimensional background turbulence, see \citet{FracBrownianMotionJonathan}. 
Although the parameter $\inerf$ is varying as the drifter changes latitude, this parameter is fixed to its mean value in each trajectory segment in \citet{sykulski2016Lagrangian}. This leaves five remaining parameters to estimate, $\{A, \lambda, B, h, \alpha \}$, in different regions of the ocean.

The model of~\eqref{driftermodel} is stationary---slowly-varying nonstationarity in the data is accounted for by windowing the data into chunks of approximately 60 inertial periods, and treating the process as locally-stationary within each window. The estimated parameters can then be aggregated spatially to quantify the heterogeneity of ocean dynamics. This method works well on relatively quiescent and stationary regions of the ocean; however this method cannot account for the rapidly-varying nonstationarity evident in Fig.~\ref{Equatorial}, and leads to model misfit and biased parameter estimates, as we shall now investigate in detail.

\begin{figure}[h!]
\centering
\includegraphics[width=0.9\textwidth]{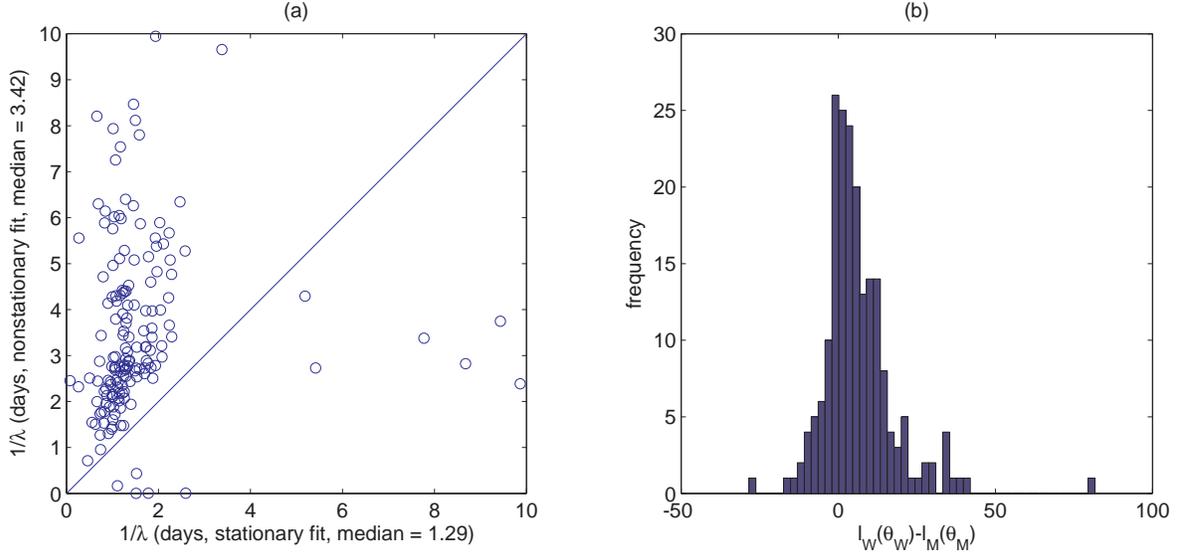}
	\caption{\label{Drifter200Analysis}(a) is a scatter plot of the damping timescale $1/\lambda$ as estimated by the stationary and nonstationary models, for each of the 200 trajectories displayed in Fig.~\ref{Equatorial}; (b) is a histogram of the difference between the log-likelihoods of the nonstationary and stationary models for the same 200 trajectories.}
\end{figure}

\subsubsection{Modulated time series modelling and estimation}\label{ss:modulateddrifters}
We now apply the methodological contributions of this paper to improve the model of~\eqref{driftermodel} for highly nonstationary time series, such as those observed in Fig.~\ref{Equatorial}(a). We do this by accounting for changes in the inertial frequency, $\inerf$, within each window of observation. We denote $\inerf(t)$ the continuous time-varying inertial frequency and $\inerf_t = \inerf(t\Delta)$ the inertial frequency value at each observed time step, $t=0,\cdots,N-1$.
The adapted version of the SDE~\eqref{eq=OU_SDE_stationary} is then given by,
\begin{equation}
	\label{eq=OU_SDE_nonstationary}
	d\widetilde{Z}_{\text{OU}}(t) = \left(-\lambda + i2\pi\inerf(t)\right)\widetilde{Z}_{\text{OU}}(t)dt + AdW(t).
\end{equation}
In analogue to the proof in \ref{proof=tvarmodulated} it is shown that the sampled process $\widetilde{Z}_{\text{OU},t}=\widetilde{Z}_{\text{OU}}(t\Delta)$ satisfies,
\begin{equation*}
	\widetilde{Z}_{\text{OU},t} = re^{i2\pi\int_{\Delta(t-1)}^{\Delta t}{\inerf(u)du}}\widetilde{Z}_{\text{OU},t-1} + \epsilon_t.
\end{equation*}
 As the inertial frequency is only observed at sampled points, we approximate the term $\int_{\Delta(t-1)}^{\Delta t}{\inerf(u)du}$ by $\Delta\inerf_t$.
Specifically, we use the model of~\eqref{eq=tvARCdef} for complex-valued time series, i.e.
\begin{equation}
	\label{eq=sampled_nonstationary_OU}
	\widetilde{Z}_{\text{OU}, t} = r e^{i2\pi\Delta\inerf_t}\widetilde{Z}_{\text{OU}, t-1} + \epsilon_t, \ \ t\geq1,
\end{equation}
where $\epsilon_t$ has the same properties as in~\eqref{eq=sampled_stationary_OU} and the transformation between the parameters $\{A, \lambda, \inerf_t\}$ of the nonstationary complex-valued OU process~\eqref{eq=OU_SDE_nonstationary} and the parameters $\{r, \sigma, \inerf_t\}$ of the nonstationary complex-valued AR(1) process~\eqref{eq=sampled_nonstationary_OU} are given by~\eqref{eq=transform_params_OU_AR}.

%The model of~\eqref{eq=sampled_nonstationary_OU} is a nonstationary complex AR(1) process. 

%\textcolor{red}{\sout{
%, and is the discrete-time analogue of the complex OU process. The complex AR(1) process can therefore be used to directly estimate $\{A,\lambda\}$ in a complex OU process with evenly sampled data, where care must be taken to correctly transform between the complex OU parameters $\{A, \lambda, \inerf_t\}$ and the complex AR(1) parameters $\{r, \sigma, \inerf_t\}$. The details of these transformations are supplied in the online material.
%}}
The required methodology has been developed in Section \ref{sec=bivariate} for bivariate (or complex-valued) time series. 
We only perform the modulation on the complex OU component in~\eqref{driftermodel}; the Mat\'ern component for the turbulent background is unchanged and is considered to be stationary in the window, as it is not in general affected by changes in $\inerf$. The two components are however observed in aggregation, and for this reason we cannot simply demodulate the observed nonstationary signal to recover a stationary signal. Instead, to jointly estimate the parameters $\{A,\lambda,B,h,\alpha\}$,
we first compute the modulating sequence, $g_t$, using~\eqref{eq=tvARcModulatingSeq} in Proposition~\ref{prop=tvarmodulated} and accounting for the temporal sample rate $\Delta$:
\begin{equation}
	\label{eq=modulationOU}
	g_t = e^{i\sum_{u=1}^t{2\pi\Delta\inerf_u}},
\end{equation}
for $t=0,\cdots,N-1$.
Then we obtain the expected periodogram of the OU component, by computing $c_g(\tau)$ according to~\eqref{eq=cgcomplex}, then $\overline{c}_{\widetilde{Z}}\sN(\tau)$, where we use the autocovariance of a stationary OU process,
\begin{equation*}
	c_{Z_{\text{OU}}}(\tau;r, \sigma) = \frac{\sigma^2}{1-r^2}r^\tau,
\end{equation*}
and Fourier transforming according to~\eqref{eq=computationOfPeriodogram}.
%Then we compute the expected periodogram of the complex AR(1) using \textcolor{red}{its autocovariance sequence as in}~\eqref{CZtilde} and then~\eqref{eq=expPeriod} 
%\textcolor{red}{with,
%\begin{equation}
	%c_{Z_{\text{OU}}}(\tau;r, \sigma) = \frac{\sigma^2}{1-r^2}r^\tau.
%\end{equation}
%}
Next, we compute the expected periodogram of the stationary Mat\'ern as outlined in~\citet{sykulski2016Lagrangian}. Note that this can also be computed from the autocovariance of a Mat\'ern using~\eqref{eq=computationOfPeriodogram}, by setting $g_t=1$ for all $t$. Finally, we additively combine the expected periodograms, i.e.
\begin{equation*}
	\overline{S}\sN(\omega;\btheta) = \sum_{\tau=-(N-1)}^{N-1}{\left[c_g(\tau)c_{Z_{\text{OU}}}(\tau) + \left(1-\frac{|\tau|}{N}\right)c_{Z_{\text{M}}}(\tau)\right]e^{-i\omega\tau}},
\end{equation*}
 and then minimize the objective function, given in~\eqref{eq=newLKHcomplex}, to obtain parameter estimates for $\{A,\lambda,B,h,\alpha\}$.

Note that the modulation of a complex-valued AR(1) process by~\eqref{eq=modulationOU} will not lead to an asymptotically stationary process, as in general we cannot expect the quantities $c_g\sN(\tau)$ to converge. However, we can see from Fig.~\ref{Equatorial}(a) that the drifters of our dataset have latitudes comprised between $\pm 20$ degrees. Therefore the terms $2\pi\Delta\inerf_t$ in~\eqref{eq=sampled_nonstationary_OU} are comprised between $\pm 0.3591$ radians, so that the conditions of Proposition~\ref{prop=tvARcCorrelationContribution} are verified. Hence the sampled inertial component is a modulated process with a significant correlation contribution, which justifies the use of our estimator~\eqref{eq=complexEstimator}. Note that this results from the latitudes of the drifters and the sampling rate used.

The assumption of Gaussianity is reasonable for modelling the velocity of instruments from the GDP as is discussed in Section 2.4 of~\citet{LaCasce20081} and references therein. To further inspect this, we tested the Gaussianity of the Fourier transform for the four velocity time series in Fig~\ref{DrifterAnalysis}. Specifically, we compared the theoretical ordered statistics of the exponential distribution to the ordered values of the normalized periodogram (normalized by the expected periodogram of the fitted Gaussian model). These results are not included in the paper for space considerations, however the code to perform this analysis can be found in the online code.

\subsubsection{Parameter estimation with equatorial drifters}
We now compare the likelihood estimates and parameter fits for the stationary model~\eqref{driftermodel}, with those for the nonstationary version of this model described in the previous subsection.
In particular, the damping timescale $1/\lambda$ is of primary interest in oceanography~\citep{elipot2010modification}.
 In Fig.~\ref{DrifterAnalysis}, we display the Whittle likelihood fits of each model to segments of data from drifters IDs 
%\#92629, \#81896, and \#71845,
\#79243, \#54656 and \#71845,
 all of which are among the trajectory segments displayed in the left-hand panel of Fig.~\ref{Equatorial}. We also include model fits to a 60-inertial period window of drifter ID\#44312, which is investigated in detail in~\citet{sykulski2016Lagrangian}, as this South Pacific drifter is from a more quiescent region of the ocean, and does not exhibit significant changes in $\inerf$. For the South Pacific drifter in Fig.~\ref{DrifterAnalysis}(d), both fits are almost equivalent (and hence are overlaid), capturing the sharp peak in inertial oscillations at approx 1.2 cycles per day. For the three equatorial drifters, the stationary model~\eqref{driftermodel} has been fit with the inertial frequency set to the average of $\omega_t^{\{f\}}$ across the window. Here in the first three cases the stationary model is a relatively poor fit to the observed time series spectra. The nonstationary modulated model, which incorporates changes in $\inerf$, is a better fit, capturing the spreading of inertial energy between the maximum and minimum values of $\inerf_t$. 

In this analysis, we have excluded frequencies higher than 0.8 cycles per day from all the likelihood fits to the equatorial drifters (the Nyquist is 6 cycles per day for this 2-hourly data), to ignore contamination from tidal energy occurring at 1 cycle per day or higher, which is not part of our stochastic model. Furthermore, we also only fit to the side of the spectrum dominated by inertial oscillations, as the model is not always seen to be a good fit on the other side of the spectrum. The modelling and inference approach is therefore semi-parametric \citep{robinson1995gaussian}. 

%In Fig. \ref{DrifterAnalysisQQplot} we provide the Quantile-Quantile plots of the ratio $\hat{S}_Z\sN(\omega)/S_Z(\omega; \hat{\btheta})$ verses the exponential distribution. Only frequencies within the estimation range are shown. If the process is Gaussian, the periodogram is expected to follow an exponential distribution. The four plots indicate a good fit, except for the tail of the distribution of drifter ID \#81896, suggesting that our assumption of Gaussianity is correct for the dataset.
%}

The significance of the misfit of the stationary model is that parameters of the model may be under- or over-estimated as the model attempts to compensate for the misfit. For example, the damping parameter of the inertial oscillations, $\lambda$, will likely be overestimated in the stationary model, as it is used to try to capture the spread of energy around $\inerf$, which is in fact mostly caused by the changing value of $\inerf$, rather than a true high value of $\lambda$.

To investigate this further, we perform the analysis with all 200 drifters shown in Fig.~\ref{Equatorial}. In Fig.~\ref{Drifter200Analysis}(a), we show a scatter plot of the estimates of $1/\lambda$, known as the damping timescale, as estimated by both models. In general, the damping timescales are larger with the nonstationary model (consistent with a smaller $\lambda$), where the median value is 3.42 days, rather than 1.3 days with the stationary model. Previous estimates of the damping timescale in the literature have not included data from the equatorial region, so while direct comparisons are not possible, the former estimates are found to be more consistent with previous estimates at higher latitudes where values of around 3 days are reported in \cite{elipot2010modification}, and values ranging from 2 to 10 days are reported in \cite{watanabe2002global}.

The nonstationary model does not require more parameters to be fitted than the stationary model; both have 5 unknown parameters. Therefore there is no need to penalize the nonstationary model using model choice or likelihood ratio tests. 
Even though the models are not nested, comparing the likelihood of the two approaches can be informative.
We can directly compare the likelihood value of each model using~\eqref{eq=whittleLKH} and \eqref{eq=newLKH}, i.e. $\ell_M(\bm{\hat\theta}_M)-\ell_W(\bm{\hat\theta}_W)$. A histogram of the difference between the likelihoods for the 200 drifters is shown in Fig.~\ref{Drifter200Analysis}(b), where positive values indicate that the likelihood of the nonstationary model is higher. Overall, the nonstationary model has a higher likelihood in 146 out of the 200 trajectories and is therefore seen to be the better model in general.

There are other regions of the global oceans, in addition to the equator, where the nonstationary methods of this paper may significantly improve parameter estimates of drifter time series. These include drifters which follow currents that traverse across different latitudes, such as the Gulf Stream or the Kuroshio. Analysis of such data is an important avenue of future investigation.

\begin{figure}[h!]
\centering
\includegraphics[width=0.9\textwidth]{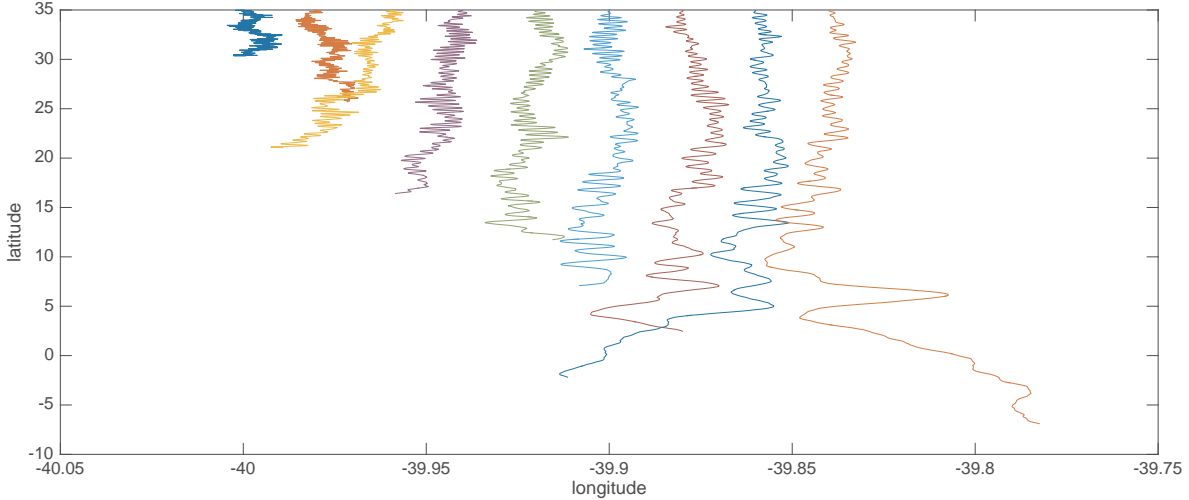}
\caption{\label{NumericalDrifters}Trajectories of 9 particles from the dynamical model, with the damping timescale set to 4 days. All particle trajectories are started at 35$^\circ$~N and 40$^\circ$~W with increasing meridional mean flow from $V=0.1$ to $V=0.9$ cm/s going from left to right ($u$ is set to zero for this example). The drifters are offset in longitude by 0.02 degrees for representation.}
\end{figure}

\subsubsection{Testing with Numerical Model Output}
\label{sec=dynamicalModel}
In this section we test the accuracy of the nonstationary modelling and parameter estimation for drifters by analysing output in a controlled setting using a dynamical model for inertial oscillations. The model propagates particles on an ocean surface forced by winds---simulated white noise in our simulations---with a fixed damping parameter, similar to the damped-slab model of~\citet{pollard1970dsr}, but uses the correct spherical dynamics for Earth from \citet{early2012forces}, so that the oscillations occur at the correct Coriolis frequency given the particle's latitude and the model remains valid at the equator. The damping timescale parameter is fixed globally {\em a priori} in the model and the goal is to see if it can be accurately estimated using parametric time series models.

The numerical model is constructed such that the particle can also be given a linear mean flow, $U+iV$. If this mean flow has a significant vertical component $V$, then the particle will cross different latitudes and the frequency of inertial oscillations will significantly change over a single analysis window.
%will change, as seen in Fig.~\ref{Equatorial} with equatorial drifters from real data. 
We display particle trajectories from the dynamical model in Fig.~\ref{NumericalDrifters}, with various realistic mean flow values, where the spherical dynamics can clearly be seen for larger latitudinal mean flow values. We observe that the particles subject to small mean flows display stationary oscillation patterns, whereas for the particles with a large latitudinal mean flow, the oscillation frequency appears to diminish as the particle approaches latitude zero.
A more complete description of the numerical model is available in the online code. 

To explore the performance of the estimation of damping time-scales, we assess the performance of the parameter estimates of our nonstationary model, by performing a Monte Carlo study based on the dynamical model described in the previous paragraph. We generate 100 trajectories, each of length 60 days and sampled every 2 hours, for a given damping timescale ($1/\lambda$) and latitudinal mean flow ($V$). We estimate the damping parameter using the stationary and nonstationary methods, in exactly the same way as with the real-world drifter data, and average the estimated damping timescales $1/\lambda$ over the 100 time series. We note that as this model has no background turbulence, then we set $B=0$ in~\eqref{driftermodel} such that there is no Mat\'ern component present.
 We then repeat this analysis over a range of realistic values for $1/\lambda$ and $V$. The average estimates of $1/\lambda$ are reported in Fig.~\ref{NumericalMatrix}. The stationary method breaks down for large mean flows and long damping timescales, with large overestimates of $\lambda$. The nonstationary method performs well across the entire range of values. We note that long damping timescales are generally harder to estimate, as $\lambda$ becomes close to zero and is estimated over relatively fewer frequencies. We have not reported mean square errors here for space considerations, but we found the parameter biases to be the main contribution to the errors, so it follows that the nonstationary method remains strongly preferable.
\begin{figure}[h!]
\centering
\includegraphics[width=0.8\textwidth]{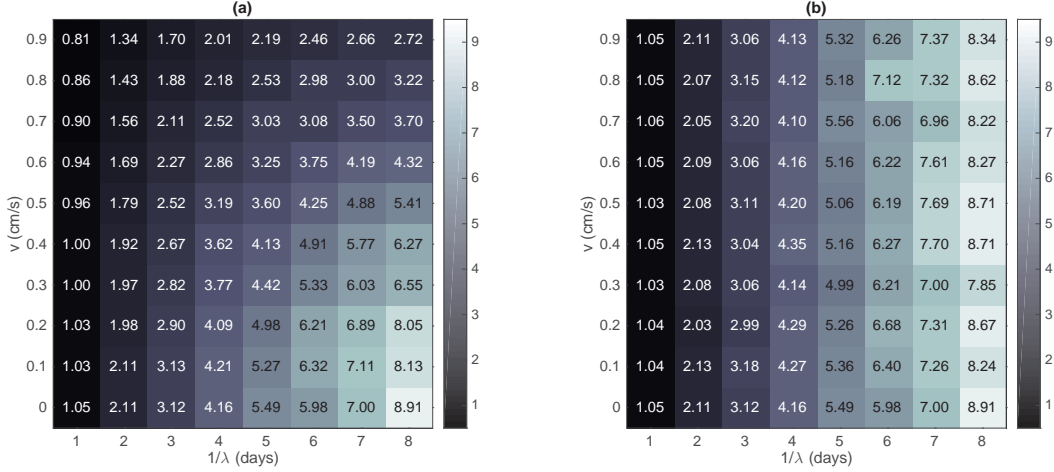}
\caption{\label{NumericalMatrix}Mean estimates of the damping timsescale $1/\lambda$ with (a) the stationary model of~\eqref{driftermodel} and (b) the nonstationary model of Section~\ref{ss:modulateddrifters}, applied to 100 realizations of the dynamical model described in Section~\ref{sec=dynamicalModel}. The experiment is performed over a grid of meridional mean flow values $v$ from 0 to 0.9 cm/s, and over a range of true damping timescales $1/\lambda$ from 1 to 8 days. The estimated damping timescale values, averaged over 100 repeat experiments, is written in each cell and shaded according to the colorbar.}
\end{figure}

\subsubsection{Testing with Stochastic Model Output}
In this section we test with purely stochastic output, which allows us to extensively compare biases, errors and computational times of the stationary and nonstationary methods in a much larger Monte Carlo study. We continue using the bivariate model of~\eqref{eq=tvARCdef} which is suitable for inertial oscillations, except this time we change $\beta_t$ according to a stochastic process. Specifically, we set as our generative mechanism for the frequencies $\beta_t$,
\begin{eqnarray}\label{eq:thetarandom}
	\beta_0 &=& \mathcal{D}(\gamma + A\epsilon_t)\\
	\beta_t &=& \mathcal{D}(\beta_{t-1} + A\epsilon_t),
\end{eqnarray}
where $\gamma\in[-\pi, \pi)$, $A>0$, $\epsilon_t$ is a standard normal white noise, and $\mathcal{D}(\cdot)$ is the bounding function defined by
\begin{equation}\label{eq:thetarandom2}
	\mathcal{D}(x) = \max\{\min(x, \gamma+\Delta), \gamma-\Delta\},
\end{equation}
where $\Delta>0$, and this choice of $\mathcal{D}(x)$ constrains $\beta_t$ in the interval $[\gamma-\Delta,\gamma+\Delta]$.
This way the frequencies $\beta_t$ are generated according to a bounded random walk, i.e. a random walk which is constrained to stay within a fixed bounded interval.
According to Proposition \ref{prop=tvARcCorrelationContribution}, if $\Delta$ is smaller than $\pi/2$, then this ensures that the modulated process belongs to the class of modulated processes with a significant correlation contribution, and our estimator~\eqref{eq=newLKHcomplex} is consistent.

In our simulations we have set $\gamma=\pi/2$, $\Delta=1$, $A=1/20$.
We simulate for a range of sample sizes ranging from $N=128$ to $N=4096$. For each sample size $N$, we independently simulate 2000 time series and estimate $\{r,\sigma\}$ for each series to report ensemble-averaged biases, errors, and computational times. The results are reported in Table~\ref{table1}. The bias and Mean Square Error (MSE) of the estimated parameters with the stationary method are seen to increase with increasing sample size. This is because the random walk of $\beta_t$ increases the range of $\beta_t$ with larger $N$, such that the nonstationarity of the time series is increasing. Conversely, the nonstationary method accounts for these rapidly changing modulating frequencies, and the bias and MSE of parameter estimates rapidly decrease with increasing $N$. The average CPU time is only around 5\% slower using the nonstationary method, as the method is still $\mathcal{O}(N\log N)$ in computational efficiency.

\begin{table}[t]
	\centering
	\caption{\label{table1}Performance of estimators with the stationary and nonstationary methods for the model of~\eqref{eq=tvARCdef} with $\beta_t$ evolving according to the bounded random walk described by~\eqref{eq:thetarandom}--\eqref{eq:thetarandom2}. The parameters are set as $r=0.8$, $\sigma=1$, $\gamma=\pi/2$, $\Delta=1$, and $A=1/20$. The results are averaged over 2000 independently generated time series for each sample size $N$. 
	The average CPU times for the optimization are given in seconds, as performed on a 2.40Ghz Intel i7-4700MQ processor (4 cores).
	}\vspace{3mm}
		\begin{tabular}{lcccccc}
				Sample size ($N$) & 128 & 256 & 512 & 1024 & 2048 & 4096\\
			\hline \hline
				\multicolumn{7}{c}{Stationary frequency domain likelihood}\\
			\hline
				Bias ($r$) &  -2.3481e-02	& -3.2400e-02 &	-4.8112e-02	& -6.9807e-02	& -9.3332e-02 &	-1.1161e-01\\

				Variance ($r$) &  1.8163e-03 &	1.0760e-03 &	1.1422e-03 &	1.5550e-03 &	1.4045e-03 &	8.2890e-04\\

				MSE ($r$) & 2.3677e-03 &	2.1258e-03 &	3.4570e-03 &	6.4280e-03 &	1.0115e-02 &	1.3286e-02\\
                                \hline
				Bias ($\sigma$) & 2.5577e-02 &	5.4988e-02 &	8.9480e-02 &	1.3241e-01 &	1.7432e-01 &	2.0651e-01\\

				Variance ($\sigma$) & 3.3898e-03 &	2.8178e-03 &	3.3471e-03 &	4.4660e-03 &	3.9885e-03 &	2.1609e-03\\

				MSE ($\sigma$)  & 4.0440e-03 &	5.8415e-03 &	1.1354e-02 &	2.1999e-02 &	3.4376e-02 &	4.4809e-02\\			
			\hline
			CPU time (sec) & 1.3083e-02 &	1.7776e-02 &	2.5743e-02 &	4.3666e-02 &	5.0948e-02 &	8.6940e-02 \\ \hline\hline
				\multicolumn{7}{c}{Nonstationary frequency domain likelihood}\\
			\hline
				Bias ($r$) & -4.6158e-03 &	-2.0129e-03 &	-1.4184e-03 &	-2.9047e-04 &	-2.6959e-04 &	8.8302e-05\\
				Variance ($r$) & 1.6508e-03 &	7.5379e-04 &	3.9819e-04 &	2.0710e-04 &	1.0674e-04 &	5.3236e-05\\
				MSE ($r$) &   1.6721e-03 &	7.5784e-04 &	4.0020e-04 &	2.0719e-04 &	1.0681e-04 &	5.3244e-05\\
			\hline
				Bias ($\sigma$) &   -1.4999e-02 &	-8.8581e-03 &	-4.4302e-03 &	-2.5292e-03 &	-1.4125e-03 &	-9.1703e-04\\

				Variance ($\sigma$) &  2.2543e-03	& 1.1989e-03 &	6.4245e-04 &	3.4775e-04 &	2.0113e-04 &	1.0759e-04\\

				MSE ($\sigma$) &  2.4793e-03 &	1.2774e-03 &	6.6208e-04 &	3.5415e-04 &	2.0312e-04 &	1.0843e-04\\
				\hline
			CPU time (sec) & 1.6814e-02 &	2.0272e-02 &	3.1397e-02 &	5.5925e-02 &	8.9997e-02 &	2.4147e-01  \\ \hline\hline
		\end{tabular}
\end{table}

Finally, we consider the case in which the modulating sequence is only unknown up to a functional form, and we must also estimate its parameters, along with the parameters of the latent process. We consider the following parametric form for $\beta_t$
\begin{equation}
	\label{eq=freqParameterModel}
	\beta_t = \gamma + \Delta\frac{2t-(N-1)}{2(N-1)},
\end{equation}
with parameters $\gamma\in[-\pi,\pi)$ and $0<\Delta<\pi$. The upper bound for $\Delta$ is chosen so that the resulting modulated process satisfies the assumptions of Proposition~\ref{prop=tvARcCorrelationContribution}. The modulated process then has a significant correlation contribution.
Therefore $\beta_t$ varies linearly from $\gamma-\frac{\Delta}{2}$ to $\gamma+\frac{\Delta}{2}$. We can then show that for all integer value $\tau$,
\begin{equation}
	c_g\sN(\tau) = \frac{\sin\left[\frac{\Delta \tau}{2(N-1)}(N-\tau)\right]}{N\sin\left[\frac{\Delta \tau}{2(N-1)}\right]}e^{\left\{i(\gamma \tau+\frac{\Delta \tau}{2(N-1)}\right\}}.
\end{equation}
This allows the kernel in~\eqref{eq=expectedAutocovSequence1} to be precomputed in $\mathcal{O}(N)$ elementary operations for all values of $\tau=0,\cdots,N-1$. This helps to speed up the computation of the expected periodogram in the likelihood for the special case of a linearly varying $\beta_t$. In this problem we have to estimate $\{\gamma,\Delta\}$ from $\beta_t$ as well as $\{r,\sigma\}$ from $Z_t$. We perform a Monte Carlo simulation with a fixed sample size of $N=512$, where we simulate 5,000 independent time series each with parameters set to  
$r=0.9$, $\sigma=10$, $\gamma=0.8$, and $\Delta=1$. We report the biases, variances and MSEs with the stationary and nonstationary methods in Table~\ref{table2}. As the stochastic process is Markovian, it is also possible to implement exact maximum likelihood in $\mathcal{O}(N)$ elementary operations for this specific problem, and we report these values in the table also. Our nonstationary inference method performs relatively close to that of exact maximum likelihood, despite the challenge of having to estimate parameters of the modulating sequence, as well as the latent process. The stationary method performs poorly, as with previous examples, as stationary modelling is not appropriate for such rapidly-varying oscillatory structure.

\begin{table}[t]
	\centering
	\caption{\label{table2}Performance of estimators with the stationary and nonstationary methods for the model of~\eqref{eq=tvARCdef} with $\beta_t$ evolving according to~\eqref{eq=freqParameterModel}. The parameters are set as $r=0.9$, $\sigma=10$, $\gamma=0.8$, and $\Delta=1$. The results are averaged over 5000 independently generated time series for each sample size $N$. N/A stands for Not Applicable.}\vspace{3mm}
		\begin{tabular}{lcccc}
				Estimated parameter & $r$ & $\sigma$ & $\gamma$ & $\Delta$\\
			\hline \hline
				\multicolumn{5}{c}{Exact likelihood}\\
			\hline
				Bias &  -1.3244e-03 &	2.1063e-02 &	2.1192e-03 &	-3.6145e-03\\

				Variance &  1.8668e-04 &	5.2130e-02 &	2.3725e-04 &	2.8323e-03\\

				MSE & 1.8844e-04 &	5.2574e-02 &	2.4174e-04 &	2.8454e-03 \\ \hline\hline
				\multicolumn{5}{c}{Stationary frequency domain likelihood}\\
			\hline
				Bias &  -1.5392e-01 &	5.2092e+00 &	2.5871e-03 & N/A\\

				Variance &  5.5052e-04 &	8.6907e-01 &	9.4628e-03& N/A\\

				MSE & 2.4241e-02 &	2.8005e+01 &	9.4695e-03   & N/A \\ \hline\hline\multicolumn{5}{c}{Nonstationary frequency domain likelihood}\\
			\hline
				Bias &  -1.7074e-03 &	6.8215e-03 &	1.1434e-03 &	-3.7092e-02\\

				Variance &  2.3975e-04 &	1.5285e-01 &	2.0803e-03 &	1.6425e-02\\

				MSE & 2.4266e-04 &	1.5290e-01 &	2.0816e-03 &	1.7801e-02 \\ \hline\hline
		\end{tabular}
\end{table}

\subsection{Missing data simulation}
\label{sec=missingSims}
In this section we show that the estimator defined in Definition \ref{def=ourEstimator} can be used for the random missing data scheme~\ref{ex=missingDataScheme} of Section \ref{sec=missing}. 
Therefore we simulate a real-valued first order autoregressive process with parameters $0\leq a <1$ and $\sigma$ according to
\begin{equation}\label{eq:missing1}
	X_t =a X_{t-1} + \epsilon_t, \ \ t\geq 1,
\end{equation}
where $X_0 \sim \mathcal{N}\left[0,\sigma^2/(1-a^2)\right]$, and $\epsilon_t$ is a Gaussian white noise process with mean zero and variance $\sigma^2$. The process $\{X_t\}$ is the latent process of interest. To account for the missing data, we generate a modulated time series $\Y_t=g_t X_t$ and assume we only observe the time series $\{\Y_t\}$, from which we estimate the parameters of the process $\{X_t\}$. The sequence $\{g_t\}$ takes its values in the set $\{0,1\}$ and is generated according to
\begin{equation*}
	g_t \sim \mathcal{B}(p_t),
\end{equation*}
where $\mathcal{B}(p)$ represents the Bernoulli distribution with parameter $p$, and where we set 
\begin{equation*}\label{eq:missing2}
	p_t = \frac{1}{2} + \frac{1}{4}\cos\left(\frac{2\pi}{10}t\right).
\end{equation*}
The observed modulating sequence $\{g_t\}$, made of zeros and ones, is clearly nonstationary as it does not admit a constant expectation. Therefore a spectral representation of the second order structure of the random modulating sequence $\{g_t\}$, as required in \citet{Dunsmuir1981b}, does not exist.
We simulate and estimate such a model for different sample sizes ranging from $N=128$ to $N=16384$. For each value of $N$, we independently simulate 2000 time series and for each time series we estimate $\{a,\sigma\}$. The outcomes of our simulation study are reported in Table~\ref{table3}. The bias, variance and mean square error rapidly decrease with increasing $N$, while the computational time only increases gradually with $N$ such that the methods are still computationally efficient for long time series. Comparing our technique with other methods from the literature is the subject of ongoing work.

\begin{table}[h]\small		\caption{\label{table3}Performance of our estimator for the missing data problem defined in~\eqref{eq:missing1}--\eqref{eq:missing2}. The unknown parameters are set as $a=0.8$, and $\sigma=1$. The results are averaged over 2000 independently generated time series for each sample size $N$.
The average CPU times for the optimization are given in seconds, as performed on a 2.40Ghz Intel i7-4700MQ processor (4 cores).
}\vspace{3mm}
	\centering
		\begin{tabular}{lccccccc}
				Sample size & 128 & 512 & 1024 & 2048 & 4096 & 8192 & 16384\\
			\hline\hline
				\multicolumn{8}{c}{Estimate of parameter $a$}\\
			\hline
				Bias  &  -2.0805e-02 &	-4.8097e-03 &	-3.0920e-04 &	-4.1710e-04 &	-2.7349e-04 &	-4.9356e-04 &	-1.5213e-04\\
				Variance  &  1.0721e-02 &	2.7114e-03 &	1.2795e-03 &	6.3100e-04 &	3.0883e-04 &	1.4380e-04 &	7.1010e-05\\
				MSE  &  1.1154e-02 &	2.7346e-03 &	1.2796e-03 &	6.3117e-04 &	3.0891e-04 &	1.4404e-04 &	7.1034e-05\\
			\hline
			\multicolumn{8}{c}{Estimate of parameter $\sigma$}\\
			\hline
				Bias  & -1.7136e-02 &	-5.4872e-03 &	-6.8949e-03 &	-2.7876e-03 &	-1.3061e-03 &	2.2880e-04 &	-1.7434e-04\\
				Variance &  3.3705e-02 &	8.7577e-03 &	4.1674e-03 &	1.9408e-03 &	9.7955e-04 &	4.1524e-04 &	2.2691e-04\\
				MSE  & 3.3999e-02 &	8.7878e-03 &	4.2150e-03 &	1.9486e-03 &	9.8125e-04 &	4.1529e-04 &	2.2694e-04\\
			\hline
			\multicolumn{8}{c}{Computational time}\\
			\hline
		CPU time (s) & 1.7901e-02 &	3.9687e-02 &	7.0515e-02 &	8.2408e-02 &	1.7624e-01 &	4.1774e-01 &	1.2138e+00\\
			\hline\hline
		\end{tabular}
\end{table}
\section{Consistency}
\label{sec=consistency}
In this section we show in Theorem~\ref{theorem=consistency} that the frequency domain estimator $\hat{\btheta}_M\sN$ (which for simplicity we denote $\hat{\btheta}\sN$ in this section) is consistent in the univariate real-valued case (extension to our class of bivariate processes follows directly). In Theorem~\ref{theorem=convergencerate} we show that this estimator converges with a $\mathcal{O}(N^{-\frac{1}{2}})$ rate. 
%For these matters, the following assumptions need to be verified.
To guarantee consistency we require the following assumptions to be satisfied:

\begin{enumerate}
\item The parameter set $\Theta\subset\R^d$ is compact with a non-null interior, and the true parameter $\theta$ lies in the interior of $\Theta$.
\item\label{assumption=1} Assume that for all $\btheta\in\Theta$, we have $\sum_{\tau\in\N}{\left |c_X(\tau;\btheta)\right|}<\infty$ (short memory) and that the functions $\theta\rightarrow c_X(\tau;\btheta)$ are continuous with respect to $\btheta$. It follows that the spectral densities are also continuous with respect to $\btheta$. We also assume that for all $\btheta\in\Theta$ and $\omega\in[-\pi,\pi]$, $S_X(\omega;\btheta)>0$. By continuity on a compact set the spectral densities $S_X(\omega;\btheta)$ are therefore bounded below in both variables by a non-zero value. For the same reason they are bounded above.
\item We assume that the spectral densities are continuously differentiable with respect to $\omega$. By continuity on a compact set the derivatives with respect to $\omega$ are bounded above, independently of $\btheta$.
%For a given continuous function $f$ defined on $[-\pi,\pi]$, we denote $V(f)$ the total variation of $f$, defined by,
%\begin{equation}
	%V(f) = \sup\left\{\sum_{i=0}^{T-1}{\left| f(\omega_{i+1})-f(\omega_i) \right|}:-\pi\leq \omega_0 < \omega_1 < \cdots < \omega_{T}\leq\pi, T\in\N\right\}.
%\end{equation}
%We assume that $\sup_{\btheta} V(S_X(\cdot, \btheta)) < \infty$, and we denote this value $V_X$.
%\item\label{assumption=2} There exists a non-negative integer $L$ such that for two non equal parameter vectors $\btheta, \btheta'\in\Omega$, we have
%that the families $\left\{c_X(\tau; \btheta);\tau=0,\cdots,L)\right\}$ and $\left\{c_X(\tau; \widetilde{\btheta});\tau=0,\cdots,L)\right\}$ are not equal.
%Note that this implies $S_X(\cdot, \btheta) \neq S_X(\cdot, \btheta')$ on a subset of positive Lebesgue measure.
\item The process $\Y_t$ is a modulated process with significant correlation contribution. We recall that this implies the existence of a finite subset $\Gamma\subset\N$ such that the mapping $\btheta\mapsto\left\{c_X(\tau):\tau\in\Gamma\right\}$ is one-to-one. We also assume that the modulating sequence $\{g_t\}$ is bounded in absolute value by some finite constant $g_{\max}>0$.
\end{enumerate}
We start with the following two lemmas which yield uniform bounds of the expected periodogram and its derivative.
%\begin{lemma}[Boundedness of the derivative of the expected periodogram]
	%\label{lemma=boundDerivative}
	%The expected periodogram admits a derivative with respect to its first parameter, which is bounded uniformely in N and the second parameter.
%\end{lemma}
%\begin{proof}
%See Appendix \ref{proof=boundDerivative}.
%\end{proof}
\begin{lemma}[Boundedness of the expected periodogram]
	\label{lemma=boundexpectedperiodogram}
	For all $\btheta\in\Theta$ and $N\in\N$, the expected periodogram $\overline{S}_{\Y}\sN(\omega;\btheta)$ is bounded below (by a positive real number) and above independently of $N$ and $\btheta$. We denote these bounds $\overline{S}_{\Y,\min}$ and $\overline{S}_{\Y,\max}$ respectively.
\end{lemma}
\begin{proof}
See Appendix \ref{proof=boundDerivative}.
\end{proof}
\begin{lemma}[Boundedness of the derivative of the expected periodogram]
\label{lemma=totalvariationnorm}
The derivative of the expected periodogram with respect to $\omega$ exists and is bounded in absolute value independently of $\btheta$ and $N$.
\end{lemma}
\begin{proof}
See appendix \ref{proof=totalvariationnorm}.
\end{proof}
In analogue to \citet{Taniguchi} for stationary processes, we introduce the following quantity,
\begin{equation*}
	D\sN\left(\bgamma, f\right) = \frac{1}{N}\sum_{\omega\in\Omega_N}\left\{ \log \overline{S}_{\Y}\sN(\omega;\bgamma) + \frac{f(\omega)}{\overline{S}_{\Y}\sN(\omega;\bgamma)} \right\},
\end{equation*}
for all positive integer $N$, $\bgamma\in\Theta$ and non-negative real-valued function $f$ defined on $\Omega_N$.
We also define
\begin{equation*}
	T\sN(f) = \arg\min_{\bgamma\in\Theta}D\sN\left(\bgamma, f\right).
\end{equation*}
This minimum for fixed $f$ is well defined since the set $\Theta$ is compact and  since the function $\bgamma\mapsto D\sN\left(\bgamma, f\right)$ is continuous. However in cases where the minimum is reached not uniquely but at multiple parameter values, 
 $T\sN(f)$ will denote any of these values, chosen arbitrarily. Note that, by the definition of our frequency domain estimator, we have $\hat{\btheta}\sN = T\sN\left(\hat{S}_{\Y}\sN(\cdot)\right)$.
We start with three lemmas that will be required in proving Theorem \ref{theorem=consistency} which establishes consistency.
\begin{lemma}
	\label{lemma=uniquenessOfMin}
	We have, for $N$ large enough, $T\sN(\overline{S}_{\Y}\sN(\omega;\btheta)) = \btheta$, uniquely.
\end{lemma}
\begin{proof}
	See Appendix~\ref{proof=uniquenessOfMin}.
\end{proof}

This shows that for all $N$ large enough, the function $\bgamma\rightarrow D\left(\bgamma, \overline{S}_{\Y}\sN(\cdot;\btheta)\right)$ reaches a global minimum at the true parameter vector $\btheta$. However because $\overline{S}_{\Y}\sN(\cdot;\btheta)$ is changing with $N$ and is not expected to converge to a given function, we need the following stronger result.

\begin{lemma}
\label{lemma=minimalValues}
If $\left\{\bgamma_N\right\}_{N\in\N}\in\Theta^\N$ is a sequence of parameter vectors such that $D\left(\bgamma_N, \overline{S}_{\Y}\sN(\cdot;\btheta)\right)-D\left(\btheta, \overline{S}_{\Y}\sN(\cdot;\btheta)\right)$ converges to zero when $N$ goes to infinity, then $\bgamma_N$ converges to $\btheta$.
\end{lemma}
\begin{proof}
See Appendix~\ref{proof=minimalValues}.
\end{proof}

We now show that the functions $D\left(\bgamma, \overline{S}_{\Y}\sN(\cdot;\btheta)\right)$ and $D\left(\bgamma, \hat{S}_{\Y}\sN(\cdot)\right)$, defined on $\Theta$, behave asymptotically \emph{in the same way}. For this, we first need the following lemma where we bound the asymptotic variance of some linear functionals of the periodogram.
\begin{lemma}
\label{lemma=boundOnVariance}
Let $\left\{a\sN(\omega): \omega\in[-\pi,\pi)\right\}_{N\in\N}$ be a family of real-valued functions, uniformly bounded by a positive real number. We have
\begin{equation*}
	\var\left\{\frac{1}{N} \sum_{\omega\in\Omega_N}{a\sN(\omega)\hat{S}_{\Y}\sN(\omega)}\right\} = \mathcal{O}\left(\frac{1}{N}\right).
\end{equation*}
\end{lemma}
\begin{proof}
See Appendix \ref{proof=boundOnVariance}
\end{proof}
Remembering that $\overline{S}_{\Y}\sN(\omega;\btheta) = \E\left\{\hat{S}_{\Y}\sN(\omega);\btheta\right\}$, we thus have 
\begin{equation}
\nonumber\sum_{\omega\in\Omega_N}{a\sN(\omega)\hat{S}_{\Y}\sN(\omega)} 
= \sum_{\omega\in\Omega_N}{a\sN(\omega)\overline{S}_{\Y}\sN(\omega;\btheta)} + \mathcal{O}_P\left(\frac{1}{\sqrt{N}}\right).
\end{equation}
We are now able to state a consistency theorem for our estimator $\hat{\btheta}\sN$.
\begin{theorem}[Consistency of the frequency domain estimator]
\label{theorem=consistency}
	We have $\hat{\btheta}\sN \overset{P}{\longrightarrow} \btheta$ in probability.
\end{theorem}
\begin{proof}
	The proof is based on~\citet{Taniguchi}. Denote $\overline{h}\sN(\bgamma;\btheta) = D\left(\bgamma, \overline{S}_{\Y}\sN(\omega;\btheta)\right)$ and
	$\hat{h}\sN(\bgamma) = D\left(\bgamma, \hat{S}_{\Y}\sN(\omega)\right)$
	defined for any $\bgamma\in\Theta$. We have,
	\begin{eqnarray*}
		\overline{h}\sN(\bgamma;\btheta) - \hat{h}\sN(\bgamma) &=&
		\frac{1}{N}\sum_{\omega\in\Omega_N}\left\{\log{\overline{S}_{\Y}\sN(\omega;\bgamma)}+\frac{\overline{S}_{\Y}\sN(\omega;\btheta)}{\overline{S}_{\Y}\sN(\omega;\bgamma)}-\log{\overline{S}_{\Y}\sN(\omega;\bgamma)}-\frac{\hat{S}_{\Y}\sN(\omega)}{\overline{S}_{\Y}\sN(\omega;\bgamma)}\right\}\\
		&=&\frac{1}{N}\sum_{\omega\in\Omega_N}{\frac{\overline{S}_{\Y}\sN(\omega;\btheta)-\hat{S}_{\Y}\sN(\omega)}{\overline{S}_{\Y}\sN(\omega;\bgamma)}}.
	\end{eqnarray*}
	We have shown in lemma \ref{lemma=boundexpectedperiodogram} that $\overline{S}_{\Y}\sN(\omega;\bgamma)$ is bounded below in both variables $\omega$ and $\gamma$ by a positive real number, independently of $N$. Therefore, making use of lemma~\ref{lemma=boundOnVariance} we have
	\begin{equation}
	\label{eq=supgoestozero}
		\sup_{\gamma\in\Omega}\left|\overline{h}\sN(\bgamma;\btheta) - \hat{h}\sN(\bgamma)\right| \stackrel{P}{\longrightarrow} 0, \ \ (N\rightarrow\infty),
	\end{equation}
	where the letter \emph{P} indicates that the convergence is in probability, as the difference is of stochastic order $N^{-\frac{1}{2}}$.
	In particular~\eqref{eq=supgoestozero} implies that 
	\[
	\left|\min_\gamma \overline{h}\sN(\bgamma;\btheta) - \min_\gamma \hat{h}\sN(\bgamma)\right| \leq \sup_{\gamma\in\Omega}\left|\overline{h}\sN(\bgamma;\btheta) - \hat{h}\sN(\bgamma)\right| \stackrel{P}{\longrightarrow}0
	\]
	i.e.
	\begin{equation}
	\label{eq=cvgzeroP1}
		\left|\overline{h}\sN\left(T\sN(\overline{S}_{\Y}\sN(\omega;\btheta));\btheta\right) -\hat{h}\sN\left(T\sN(\hat{S}_{\Y}\sN(\omega))\right) \right| \stackrel{P}{\longrightarrow} 0.
	\end{equation}	
 Relation (\ref{eq=supgoestozero}) also implies that
	\begin{equation}
	\label{eq=cvgzeroP2}
	\left| \overline{h}\sN\left(T\sN(\hat{S}_{\Y}\sN(\omega));\btheta\right)-\hat{h}\sN\left(T\sN(\hat{S}_{\Y}\sN(\omega))\right)  \right| \stackrel{P}{\longrightarrow} 0 ,
	\end{equation}
	so that using the triangle inequality,~\eqref{eq=cvgzeroP1} and~\eqref{eq=cvgzeroP2}, we get,
	\begin{equation*}
		\left| \overline{h}\sN\left(T\sN(\hat{S}_{\Y}\sN(\omega));\btheta\right)  - \overline{h}\sN\left(T\sN(\overline{S}_{\Y}\sN(\omega;\btheta));\btheta\right) \right| \stackrel{P}{\longrightarrow} 0.
	\end{equation*}	
	We then obtain the stated theorem making use of Lemma \ref{lemma=minimalValues}.
\end{proof}

We now study the convergence rate of our frequency domain estimator. For this we first need the following two lemmas. Although the Hessian matrix of the likelihood
is not expected to converge for modulated processes with a significant correlation contribution, we can show that its norm is bounded below 
by a positive real number.
For this we need to strengthen the assumption of significant correlation contribution. Assuming that the spectral densities of the latent process are twice continuously differentiable with respect to $\btheta$, we assume that the Jacobian determinant of the mapping $\btheta\mapsto \left[c_X(\tau;\btheta):\tau\in\Gamma\right]^T$ taken at the true parameter value $\btheta$, i.e. the determinant of the matrix with elements $\frac{\partial c_X(\tau_i;\btheta)}{\partial\btheta_j}$ (with $\Gamma = \{\tau_1, \tau_2, \cdots, \tau_d\}$ here), is non-zero.

\begin{lemma}
\label{lemma=boundbelowC}
	Let $\bold{U}_1, \cdots, \bold{U}_d$ a family of vectors of $\R^d$ with rank $d$. Let $\alpha_1, \cdots, \alpha_d$ be positive real numbers. There exists a positive constant $C>0$ such that for all $\bold{V}\in\R^d$, 
	\begin{equation}
		\sum_{i=1}^d{\alpha_i^2\left(\bold{U}_i^T\bold{V}\right)^2} \geq C \left\|\bold{V}\right\|_2^2,
	\end{equation}
	where $\|\cdot\|_2$ denotes the Euclidean norm on $\R^N$. 
\end{lemma}
\begin{proof}
See Appendix~\ref{proof=boundbelowC}.
\end{proof}

\begin{lemma}
\label{lemma=7}
We have,
\begin{equation}
	\frac{\partial l_M\sN}{\partial\theta_i}(\btheta) = \mathcal{O}_P\left(\frac{1}{\sqrt{N}}\right).
\end{equation}
The Hessian matrix of the function $l_M(\btheta)$ satisfies
\begin{equation}
	H(\btheta) = \mathcal{I}(\btheta) + \mathcal{O}_P\left(\frac{1}{\sqrt{N}}\right),
\end{equation}
where the matrix norm of $\mathcal{I}(\btheta)$ is bounded below by a positive value, independently of $N$.
\end{lemma}
\begin{proof}
See Appendix~\ref{proof=7}.
\end{proof}

\begin{theorem}[Convergence rate]
\label{theorem=convergencerate}
	We have $\hat{\btheta} = \btheta + \mathcal{O}_P\left(\frac{1}{\sqrt{N}}\right)$.
\end{theorem}
\begin{proof}
We have, by Taylor expansion with Lagrange form of the remainder term,
\begin{equation*}
	\nabla l_M(\hat{\btheta}) = \bold{0} = \nabla l_M(\btheta) + H(\widetilde{\btheta})(\hat{\btheta}-\btheta),
\end{equation*}
where $\widetilde{\btheta}$ lies between $\hat{\btheta}$ and $\btheta$. Therefore,
\begin{equation}
	\hat{\btheta}-\btheta = - H(\widetilde{\btheta})^{-1}\nabla l_M(\btheta).
\end{equation}
We have shown that $\hat{\btheta}$ converges in probability to $\btheta$. By continuity of the Hessian, and using the results of Lemma~\ref{lemma=7}, we obtain
\begin{equation}
	\hat{\btheta}-\btheta = - \left[\mathcal{I}(\btheta) + \mathcal{O}_P\left(\frac{1}{\sqrt{N}}\right) + o_P(1)\right]^{-1}\mathcal{O}_P\left(\frac{1}{\sqrt{N}}\right) = \mathcal{O}_P\left(\frac{1}{\sqrt{N}}\right).
\end{equation}
This concludes the proof.
\end{proof}

\section{Conclusion}
The well-established theory for the analysis of stationary time series is often in contradiction with real-world data applications. This is because most real time series are nonstationary. Nonstationary observations have required statisticians to develop new  models and more broadly new generating mechanisms. Among the large class of nonstationary models, uniform modulation of time series is an easy way to create nonstationarity, and presents all the advantages of a simple mechanism for the time-varying second order structure of a process.
Modulation has already been used to account for missing data when analysing stationary time series, as well as gentle time variation. In fact, if the modulation is slow, regular theory for locally stationary time series applies. Despite its popularity as a modelling tool, the concept of modulation, when variation can be moderate to rapid, is very poorly understood. We have in this paper shown how modulation of time series can account for much more rapid changes of a time series model.~\citet{JTSA:JTSA12034} already abandoned the assumption of smoothness in time of the time-varying spectral density~\citep{Dahlhaus1997}. However, the class of modulated processes with a significant correlation contribution is one of few instances of nonstationary models where more data in time results in more accurate estimates, and asymptotic consistency under standard assumptions for the latent stationary process.

As we have generalized modulation beyond the assumptions where it is known that models can be estimated, the question naturally arises, as to what types of modulation still permit parameter estimation. 
Key to our understanding of modulated processes is the definition of modulated processes with a significant correlation contribution (see Definition \ref{def=univariateSignificantCorrel}), which generalizes the classical concept of asymptotically stationary modulated processes, and corresponds to our main modelling innovation. We require that the sample autocorrelations of the modulating sequences be asymptotically bounded below, so that the information in the autocorrelation of the process does not fade in the observed process. The interpretation of this requirement is that there must be sufficient support in the autocorrelation to retain the information in the modulation.

With this new model class we can implement estimation directly in the Fourier domain, directly after transforming the data from the temporal domain. Estimation is still possible in ${\cal O}\left(N\log(N)\right)$ computational effort, and the further required conditions to ensure consistency were studied.
Most real-world data sets are aggregations of heterogeneous components. To fully show the promise of our newly proposed procedure, we show how estimation is still possible in the setting of unobserved components models, where different types of processes are superimposed. 
Real-world data from the Global Drifter Program show its relevance for understanding surface flow measurements at the equator---a challenging region for studying inertial oscillations—where the power of the new method shows that despite rapid modulation we can still uncover the generating mechanism of the process.

 There are a number of questions still remaining in our understanding of modulation. We have extended the regimes when estimation is possible, but do not know when an estimable process tips into one from which no information can be recovered. 
 By introducing a new class of models, many new questions can both be posited and answered, especially as most sources of real-world data show aggregations of components, all obeying different generation mechanisms.
\footnotesize
\setlength{\bibsep}{0pt plus 0.2ex}
\bibliographystyle{apalike}
\bibliography{arthur}

\appendix
\section{Appendix}

\subsection{Proof of Proposition \ref{prop=stationaryModulatedProcesses}}
\label{proof=stationaryModulatedProcesses}
\begin{proof}
We distinguish the case where $\{Z_t\}$ is a white noise process (the covariance is zero everywhere except for lag zero) from the case where $\{Z_t\}$ is not a white noise process. 
\begin{enumerate}
		\item\noindent Assume $\{Z_t\}$ is a white noise process.
		\begin{itemize}
			\item[$\rightarrow$] Assume $\{\widetilde{Z}_t\}$ is a stationary process.
			Being stationary, it has a constant variance and therefore the modulating sequence must have a constant modulus.
			\item[$\leftarrow$] Conversely, if $\{g_t\}$ has a constant modulus $\{\widetilde{Z}_t\}$ is stationary and is a white noise process.
		\end{itemize}
		
		\item\noindent Assume $\{Z_t\}$ is not a white noise process. The set $\{\tau\in\N^* :|c_Z(\tau;\boldsymbol{\theta})|>0\}$ is therefore not empty, so 
		$\mu = \gcd\{\tau\in\N^* :|c_Z(\tau;\boldsymbol{\theta})|>0\}$ is well defined.
		\begin{itemize} 
		\item[$\rightarrow$] Assume $\{\widetilde{Z}_t\}$ is stationary.
		Then it must have a constant variance, so there must exists a real number $a\geq 0$ such that $\rho_t = a\ \forall t\in\N$.
		Leaving aside the trivial case in which $\{g_t\}$ is zero everywhere, let $t_1, t_2$ be two natural integers.
		We have

	\begin{equation}
		\nonumber c_{\widetilde{Z}}(t_1,t_2;\boldsymbol{\theta}) 
		= g_{t_1}^* g_{t_2} c_Z(t_2-t_1;\boldsymbol{\theta}) 
		= a^2e^{i(\phi_{t_2}-\phi_{t_1})}c_Z(t_2-t_1;\boldsymbol{\theta}). 
	\end{equation}
	If $c_Z(t_2-t_1;\boldsymbol{\theta})\neq 0$ then

	\begin{equation}
		\nonumber e^{i(\phi_{t_2}-\phi_{t_1})} = 
		\frac{c_{\widetilde{Z}}(s,t;\boldsymbol{\theta})}{a^2 c_Z(t_2-t_1;\boldsymbol{\theta})},
	\end{equation}
	which leads to

	\begin{equation}
		\nonumber \phi_{t_2}-\phi_{t_1} = 
		\arg\left\{\frac{c_{\widetilde{Z}}(t_1,t_2;\boldsymbol{\theta})}{a^2 c_Z(t_2-t_1;\boldsymbol{\theta})}\right\} \mod 2\pi,
	\end{equation}
	where the equality is true up to a multiple of $2\pi$, which we indicate by the use of the notation $\mod 2\pi$.
	Since $\{\widetilde{Z}_t\}$ is assumed stationary, there exists a function $\zeta$, 
	defined on $\left\{\tau\in\N:c_Z(\tau;\boldsymbol{\theta})\neq0\right\}$, such that 

	\begin{equation}
		\nonumber \arg\left\{\frac{c_{\widetilde{Z}}(t_1,t_2;\boldsymbol{\theta})}{a^2 c_Z(t_2-t_1;\boldsymbol{\theta})}\right\} = 
		\zeta(t_2-t_1) \mod 2\pi, \ \forall t_1,t_2\in\N.
	\end{equation}
	Therefore

	\begin{equation}
		\nonumber \phi_{t_2}-\phi_{t_1} = \zeta(t_2-t_1) \mod 2\pi.
	\end{equation}
	Now let $t\in\N$ be any natural integer and write $t =\mu q + r$ where $0\leq r<\mu$ and $q\in\N$ are uniquely defined as the remainder and quotient of the Euclidian division of $t$ by $\mu$. 

	\begin{equation*}
		 \phi_t = \sum_{k=0}^{q-1}{(\phi_{r+(k+1)\mu} - \phi_{r+k\mu})} + \phi_r
		 = \sum_{k=0}^{q-1}{\zeta(\mu)} + \phi_r \mod 2\pi
		 = q\zeta(\mu) + \phi_r \mod 2\pi.
	\end{equation*}

	Letting $\gamma = \zeta(\mu)$ we obtain,

	\begin{equation}
		\nonumber \phi_t = \gamma \left\lfloor{\frac{t}{\mu}}\right\rfloor + \phi_{t\bmod\mu} \mod 2\pi.
	\end{equation}
	
	\item[$\leftarrow$] 
	Conversely assume there exists two constants $\gamma\in\R$ and $a\geq0$ such that for all $t\in\N$,

	\begin{eqnarray}
		\nonumber \rho_t &=& a,\\
		\nonumber \phi_t &=& \phi_{t\bmod\mu} + \gamma \left\lfloor{\frac{t}{\mu}}\right\rfloor \mod 2\pi.
	\end{eqnarray}

	Let $t,\tau$ be two natural integers.
	We have:

	\begin{equation}
		\nonumber c_{\widetilde{Z}}(t,t+\tau;\boldsymbol{\theta}) = 
		g_t^* g_{t+\tau} c_Z(\tau;\boldsymbol{\theta}) = 
		a^2e^{i(\phi_{t+\tau}-\phi_t)}c_Z(\tau;\boldsymbol{\theta}). 
	\end{equation}
	If $c_Z(\tau;\boldsymbol{\theta}) = 0$ then $c_{\widetilde{Z}}=(t,t+\tau;\boldsymbol{\theta})=0$ which does not depend on $t$.
	Otherwise, $\tau$ is a multiple of $\mu$ by definition of $\mu$.
	Therefore there exists an integer $q$ such that $\tau=q\mu$, and $(t+\tau)\bmod\mu = t \bmod\mu$.
	Finally,

	\begin{eqnarray}
		\nonumber \phi_{t+\tau}-\phi_t &=& \phi_{(t+\tau)\bmod \mu} 
		+ \gamma \left\lfloor{\frac{t+\tau}{\mu}}\right\rfloor - \phi_{t\bmod \mu} 
		- \gamma \left\lfloor{\frac{t}{\mu}}\right\rfloor \mod 2\pi\\
		\nonumber &=& \gamma \left\lfloor{\frac{t}{\mu}+q}\right\rfloor - 
		\gamma \left\lfloor{\frac{t}{\mu}}\right\rfloor \ \mod 2\pi
		 = \gamma \left ( \left\lfloor{\frac{t}{\mu}}\right\rfloor + q - \left\lfloor{\frac{t}{\mu}}\right\rfloor \right) \ \mod 2\pi\\
		\nonumber &=& \gamma q \ \mod 2\pi,
	\end{eqnarray}
	where we have used the fact that $\left\lfloor{\frac{t}{\mu}+q}\right\rfloor = \left\lfloor{\frac{t}{\mu}}\right\rfloor + q$ as $q$ is an integer. Again the obtained quantity does not depend on $t$.
	Therefore $c_{\widetilde{Z}}(t,t+\tau;\boldsymbol{\theta})$ does not depend on $t$ but only on the lag $\tau$.
	This proves that $\{\widetilde{Z}_t\}$ is stationary with autocovariance sequence
	$
		c_{\widetilde{Z}}(\tau;\boldsymbol{\theta}) = a^2e^{i\gamma\frac{\tau}{\mu}}c_Z(\tau).
	$
	As for the spectral density of the resulting stationary modulated process $\{\widetilde{Z}_t\}$ we have,
	\begin{eqnarray}
		\nonumber S_{\widetilde{Z}} (\omega;\boldsymbol{\theta}) &=& \sum_{\tau=-\infty}^{\infty}{c_{\widetilde{Z}}(\tau;\boldsymbol{\theta})e^{-i\omega \tau}}
		\nonumber = \sum_{\tau=-\infty}^{\infty}{a^2 c_Z(\tau;\boldsymbol{\theta})e^{-i(\omega \tau - \gamma\frac{\tau}{\mu})}}
		\nonumber = \sum_{\tau=-\infty}^{\infty}{a^2 c_Z(\tau;\boldsymbol{\theta})e^{-i(\omega - \frac{\gamma}{\mu})\tau}}\\
		\nonumber &=& a^2\sum_{\tau=-\infty}^{\infty}{c_Z(\tau;\boldsymbol{\theta})e^{-i(\omega - \frac{\gamma}{\mu})\tau}}
		\nonumber =  a^2S_Z\left(\omega-\frac{\gamma}{\mu}\right).
	\end{eqnarray}
	\end{itemize}
\end{enumerate}
This concludes the proof. Note that for a real-valued process this shift would be impossible as the spectral density has to retain symmetry. As both ${Z_t}$ and ${\widetilde{Z}_t}$ are complex-valued, this is not a concern. 
\end{proof}

\subsection{Proof of Proposition \ref{prop=tvarmodulated}}
\label{proof=tvarmodulated}
\begin{proof}
	Let us define the complex-valued stochastic process $\{Z_t\}$ according to

	\begin{equation}
		\nonumber Z_t = e^{-i\sum_{u=1}^{t}{\beta_u}}\widetilde{Z}_t, \ \ t=0,1,2,\cdots .
	\end{equation}
	By applying the definition of the process $\{\widetilde{Z}_t\}$ one can determine the following relationship, for all $t\geq1$,

	\begin{equation*}
		 Z_t = e^{-i\sum_{u=1}^t{\beta_u}}\widetilde{Z}_t
		 \ =\ e^{-i\beta_t}e^{-i\sum_{u=1}^{t-1}{\beta_u}}\widetilde{Z}_t
		 \ = \ e^{-i\beta_t}e^{-i\sum_{u=1}^{t-1}{\beta_u}}(re^{i\beta_t}\widetilde{Z}_{t-1}+\epsilon_t)
		 \ =\ re^{-i\sum_{u=1}^{t-1}{\beta_u}}\widetilde{Z}_{t-1} +\epsilon_t',
	\end{equation*}
	and finally
	$
		\nonumber Z_t = r Z_{t-1} + \epsilon_t', \ \ t\geq 1,
	$
	where $\epsilon_t' = e^{-i\sum_{u=1}^t{\beta_u}}\epsilon_t, \forall t\in\N$ 
	has the same distribution as $\epsilon_t$, as we have assumed that the complex-valued white noise process
	$\epsilon_t$ has variance $\sigma^2$ and independent real and imaginary parts.
	Therefore the process $Z_t$ is a first-order complex-valued autoregressive process with stationary parameters. It
	is stationary if and only if $\var\{Z_0\} = \frac{\sigma^2}{1-r^2}$.
	Since $Z_0=\widetilde{Z}_0$, it follows that $\var\{Z_0\} = \var\{\widetilde{Z}_0\}$.
	Thus if $\var\{\widetilde{Z}_0\} = \frac{\sigma^2}{1-r^2}$, the process $\{Z_t\}$ is stationary. The fact that the process $\{Z_t\}$ is proper stems from the fact that
	the innovations $\{\epsilon_t\}$ as well as the random variable $\widetilde{Z}_0$ are proper, as using the following relation,
		$Z_{t} = r^t Z_0 + \sum_{j=1}^t{r^{t-j}\epsilon'_{j}}$,
	we obtain for all $t,\tau\in\N$, 
	$
		\E\left\{Z_{t}Z_{t+\tau}\right\} = 0.
	$
	This shows how the proposed process is generated by the stated mechanism of modulation as claimed in the proposition.
\end{proof}

\subsection{Proof of Proposition \ref{prop=tvARcCorrelationContribution}}
\label{proof=tvARcCorrelationContribution}
\begin{proof}
	Let $\Gamma = \{0,1\}$. We show that conditions 1 and 2 given in Definition \ref{def=univariateSignificantCorrel} are verified.
	\begin{enumerate}
	\item The function $(r,\sigma)\mapsto \left\{c_Z\left[\tau;(r,\sigma)\right]:\tau\in\Gamma\right\}$ is one-to-one.
	\item According to Proposition \ref{prop=tvarmodulated}, there exists a stationary proper complex-valued process $Z_t$ such that $\widetilde{Z}_t = g_tZ_t$, where 
	$g_t = e^{i\sum_{u=1}^{t}{\beta_u}}$, i.e $\widetilde{Z}_t$ is a modulated process. The autocovariance sequence of the process $Z_t$ is given by 
	\begin{equation*}
	c_Z(\tau) = \frac{\sigma^2}{1-r^2}r^\tau, \tau\in\Z,
\end{equation*}
and we observe that the function $(r, \sigma)\mapsto \left(c_X(0), c_X(1)\right)$ is one-to-one.
	Let $L$ be the largest positive (i.e. greater than or equal to 1) integer such that $\Delta\leq\frac{\pi}{2L}$. This is well defined as we have assumed $0\leq\Delta\leq\frac{\pi}{2}$.
	Fix an integer lag value $0\leq\tau\leq L$.
	We have
	\begin{eqnarray*}
		\left|\frac{1}{N}\sum_{t=0}^{N-1-\tau}{g_t^* g_{t+\tau}}\right|&=& 
		\left|\frac{1}{N}\sum_{t=0}^{N-1-\tau}{e^{i\sum_{u=t}^{t+\tau-1}{\beta_u}}}\right|
		\ =\ \left|\frac{1}{N}\sum_{t=0}^{N-1-\tau}{e^{i\sum_{u=t}^{t+\tau-1}\left(\Xi+\beta_u-\Xi\right)}}\right|\\
		&=& \left|\frac{1}{N}e^{i\tau\Xi}\sum_{t=0}^{N-1-\tau}{e^{i\sum_{u=t}^{t+\tau-1}\left(\beta_u-\Xi\right)}}\right|
		\ \geq\ \frac{1}{N}\left|\Re\left\{\sum_{t=0}^{N-1-\tau}{e^{i\sum_{u=t}^{t+\tau-1}\left(\beta_u-\Xi\right)}}\right\}\right|\\
		&=& \frac{1}{N}\left|\sum_{t=0}^{N-1-\tau}\cos\left\{\sum_{u=t}^{t+\tau-1}\left(\beta_u-\Xi\right)\right\}\right|.
	\end{eqnarray*}
	Using the triangle inequality it follows $\left|\sum_{u=t}^{t+\tau-1}{\beta_u-\Xi}\right| \leq \sum_{u=t}^{t+\tau-1}{|\beta_u-\Xi|} \leq \tau\Delta$, 
	the fact that $\tau\Delta < \frac{\pi}{2}$ by assumption, and that the cosine function is 
	decreasing on the interval $[0,\frac{\pi}{2}]$ we obtain
	\[
		\left|\frac{1}{N}\sum_{t=0}^{N-1-\tau}{g_t^* g_{t+\tau}}\right|\geq \frac{1}{N}\sum_{t=0}^{N-1-\tau}\cos(\tau\Delta)
		= (1-\frac{\tau}{N})\cos(\tau\Delta) \stackrel{N\rightarrow\infty}{\rightarrow}  \cos(\tau\Delta) > 0,
	\]
	as $0\leq\tau\beta<\tau\frac{\pi}{2L}\leq \frac{\pi}{2}$.
	The above converges to a non-zero value as $N$ goes to infinity, so that $\liminf\limits_{N\rightarrow\infty}\left|c_g\sN(\tau)\right| > 0$. It is true in particular for $\tau\in\Gamma$. 
\end{enumerate}
This shows that the process $\widetilde{Z}_t$ is a modulated process with a significant correlation contribution.
\end{proof}

\subsection{Proof of Lemma \ref{lemma=boundexpectedperiodogram}}
\label{proof=boundDerivative}
\begin{proof}
We denote $S_{X,\max} = \max_{\btheta,\omega}S_X(\omega;\btheta)$ and $S_{X,\min} = \min_{\btheta,\omega}S_X(\omega;\btheta)$.
	\begin{enumerate}
	\item We first show the existence of the upper bound.
	According to Proposition \ref{prop=freqPropSbar} the expected periodogram can be expressed, for $\omega\in[-\pi,\pi]$, $\btheta\in\Theta$ and $N\in\N$, by
	\begin{equation*}
	\overline{S}_{\Y}\sN(\omega;\boldsymbol{\theta}) =
		2\pi\int_{-\pi}^\pi{S_X(\omega-\lambda)S_g\sN(\lambda)d\lambda}.
	\end{equation*}
	Therefore,
	\begin{equation*}
		\overline{S}_{\Y}\sN(\omega;\boldsymbol{\theta}) \leq 2\pi S_{X,\max} \int_{-\pi}^\pi{S_g\sN(\lambda)d\lambda}
		= S_{X,\max} \frac{1}{N}\sum_{t=0}^{N-1}{\left|g_t\right|^2},
	\end{equation*}
	by Parseval equality, and finally,
	\begin{equation*}
	\overline{S}_{\Y}\sN(\omega;\boldsymbol{\theta}) \leq g^2_{\max} S_{X,\max},
	\end{equation*}
	and this by assumption is finite.
	\item Similarly, we show the existence of a lower bound. According to the assumption of a modulated process with significant correlation contribution, there exists a non-negative integer $\tau\in\Gamma$ and a positive real number $\alpha_\tau$ such that for $N$ large enough, $c_g\sN(\tau)\geq\alpha_\tau$. Then,
	\begin{eqnarray*}
		\overline{S}_{\Y}\sN(\omega;\boldsymbol{\theta}) &\geq& 2\pi S_{X,\min} \int_{-\pi}^\pi{S_g\sN(\lambda)d\lambda}
		\ =\ S_{X,\min} \frac{1}{N}\sum_{t=0}^{N-1}{\left|g_t\right|^2}\\
		&\geq& S_{X,\min} \frac{1}{N}\sum_{t=0}^{N-\tau-1}{\left|g_t\right|^2}
		\ \geq\ S_{X,\min} \frac{1}{N}\sqrt{\sum_{t=0}^{N-\tau-1}{\left|g_t\right|^2}}\sqrt{\sum_{t=\tau}^{N-1}{\left|g_t\right|^2}}\\
		&\geq& S_{X,\min} \frac{1}{N}\left|\sum_{t=0}^{N-\tau-1}{g_t^* g_{t+\tau}} \right|,
	\end{eqnarray*}
	by Cauchy-Schwartz inequality.
	Hence we get for $N$ large enough
	$
		\overline{S}_{\Y}\sN(\omega;\boldsymbol{\theta}) \geq \alpha_\tau S_{X,\min}.
	$
	This proves the stated result.
	\end{enumerate}
\end{proof}

\subsection{Proof of Lemma \ref{lemma=totalvariationnorm}}
\label{proof=totalvariationnorm}
\begin{proof}
We have for all $\omega\in[-\pi,\pi]$, $\btheta\in\Theta$ and $N\in\N$ that the form of $\overline{S}_{\Y}\sN(\omega;\boldsymbol{\theta})$ is given by 
	\begin{equation*}
	\overline{S}_{\Y}\sN(\omega;\boldsymbol{\theta}) =
		2\pi\int_{-\pi}^\pi{S_X(\omega-\lambda;\btheta)S_g\sN(\lambda)d\lambda}.
	\end{equation*}
We obtain (where the inversion of differentiation and integration is a consequence of the differentiability of the functions $\omega\rightarrow S_X(\omega;\btheta)$ and the fact that the spectral densities are bounded above),
\begin{equation*}
	\left|\frac{\partial \overline{S}_{\Y}\sN}{\partial\omega}(\omega;\btheta)\right|
	= 2\pi\left|\int_{-\pi}^\pi{\frac{\partial S_X}{\partial\omega}(\omega-\lambda;\btheta)S_g\sN(\lambda)d\lambda}\right|
	\ \leq \ 2\pi\max_{\lambda, \btheta}\left\{\left|\frac{\partial S_X}{\partial\omega}(\lambda;\btheta)\right|\right\}\int_{-\pi}^\pi{S_g\sN(\lambda)d\lambda}
	\ \leq \ g_{\max}^2\max_{\omega, \btheta}\left\{\left|\frac{\partial S_X}{\partial\omega}(\omega;\btheta)\right|\right\},
\end{equation*}
which concludes the proof.
\end{proof}

\subsection{Proof of Lemma \ref{lemma=uniquenessOfMin}}
\label{proof=uniquenessOfMin}
\begin{proof}
We will use repeatedly the fact that the function $x\rightarrow x-\log x$, defined on the set of positive real numbers, admits a global minimum for $x=1$ where it takes value 1. It is an increasing function on $(1,\infty)$ and decreasing on $(0,1)$. This is easily seen by studying the derivative.
Now let $N$ be a natural integer. We have for all $\bgamma\in\Theta$
\begin{eqnarray*}
	D(\bgamma, \overline{S}_{\Y}\sN(\omega;\btheta)) &=& \frac{1}{N}\sum_{\omega\in\Omega_N}
		\left\{ 
			\log \overline{S}_{\Y}\sN(\omega;\bgamma) + 
			\frac{\overline{S}_{\Y}\sN(\omega;\btheta)}{\overline{S}_{\Y}\sN(\omega;\bgamma)} 
		\right\} 
		=\frac{1}{N} \sum_{\omega\in\Omega_N}
		\left\{
		\log \overline{S}_{\Y}\sN(\omega;\btheta) + \overbrace{
		\frac{\overline{S}_{\Y}\sN(\omega;\btheta)}{\overline{S}_{\Y}\sN(\omega;\bgamma)} - 
		\log\left[  \frac{\overline{S}_{\Y}\sN(\omega;\btheta)}{\overline{S}_{\Y}\sN(\omega;\bgamma)}\right]
		}^{\geq 1}
		\right\}\\
		&\geq& \frac{1}{N}\sum_{\omega\in\Omega_N}
		\left\{
		\log \overline{S}_{\Y}\sN(\omega;\btheta)
		 +1\right\},
	\end{eqnarray*}
	where we have an equality if and only if $\overline{S}_{\Y}\sN(\omega;\btheta) = \overline{S}_{\Y}\sN(\omega;\bgamma)$ for all $\omega\in\Omega_N$, which for $N$ large enough is equivalent to $\bgamma=\btheta$ according to Proposition~\ref{prop=identifiabilityViaPeriodogram}.
\end{proof}

\subsection{Proof of Lemma \ref{lemma=minimalValues}}
\label{proof=minimalValues}
\begin{proof}
We prove this in three steps.
\begin{enumerate}
\item
We have for a positive integer $N$,
	\begin{equation}
	\label{eq=d-d}
		D\left(\bgamma_N, \overline{S}_{\Y}\sN(\cdot;\btheta)\right)-D\left(\btheta, \overline{S}_{\Y}\sN(\cdot;\btheta)\right) =
		\frac{1}{N}\sum_{\omega\in\Omega_N}
		\left\{
			\frac{\overline{S}_{\Y}\sN(\omega;\btheta)}{\overline{S}_{\Y}\sN(\omega;\bgamma_N)}
			- \log\frac{\overline{S}_{\Y}\sN(\omega;\btheta)}{\overline{S}_{\Y}\sN(\omega;\bgamma_N)}
			-1
		\right\}.
	\end{equation}
	By assumption, this converges to zero as $N$ goes to infinity.
For any integer positive $\tau$ smaller than $N$ we can write,
\begin{equation}
	\overline{c}_{\Y}\sN(\tau;\bgamma_N)-\overline{c}_{\Y}\sN(\tau;\btheta) = 
	\frac{1}{2\pi}
	\int_{-\pi}^{\pi}
	\left(
		\overline{S}_{\Y}\sN(\omega;\bgamma_N) - \overline{S}_{\Y}\sN(\omega;\btheta)
	\right)
	e^{i\omega\tau}d\omega,
\end{equation}
so we have the following bound,
\begin{equation}
\label{eq=cy-cy}
	\left|\overline{c}_{\Y}\sN(\tau;\bgamma_N)-\overline{c}_{\Y}\sN(\tau;\btheta)\right| \leq
	\frac{1}{2\pi}
	\int_{-\pi}^{\pi}{
	\left|
		\overline{S}_{\Y}\sN(\omega;\bgamma_N) - \overline{S}_{\Y}\sN(\omega;\btheta)
	\right|
	d\omega}.
\end{equation}
\item
Now we assume, with the intent to arrive at a contradiction, that this quantity does not converge to zero. Then there exists an increasing function $\phi(N)$, defined on the set of non-negative 
integers and taking values in the set of non-negative integers and $\epsilon>0$ such that
\begin{equation}
	\left|\overline{c}_{\Y}\sfN(\tau;\bgamma\ufN)-\overline{c}_{\Y}\sfN(\tau;\btheta)\right| \geq\epsilon, \ \forall N\in\N. 
\end{equation}
Fix $N\in\N$. Denote $M$ the upper bound (independent of $N$) of the integrands in (\ref{eq=cy-cy}) using lemma \ref{lemma=boundexpectedperiodogram}. 
Let $B_{\phi(N)}\subset[-\pi,\pi]$ be the inverse image of $[\epsilon/2,\infty)$ by the function $\omega\mapsto \left|\overline{S}_{\Y}\sfN(\omega;\bgamma\ufN) -\overline{S}_{\Y}\sfN(\omega;\btheta)\right|$. Let $\lambda_{\phi(N)}$ be the Lebesgue measure of the Borel set $B_{\phi(N)}$. We have,
\begin{equation}
\epsilon\leq 
	\frac{1}{2\pi}
		\int_{-\pi}^{\pi}{\left|
		\overline{S}_{\Y}\sfN(\omega;\bgamma\ufN) - \overline{S}_{\Y}\sfN(\omega;\btheta)
	\right|d\omega} \leq \frac{\lambda_{\phi(N)}}{2\pi}M + \frac{2\pi-\lambda_{\phi(N)}}{2\pi}\frac{\epsilon}{2},
\end{equation}
and therefore
\begin{equation}
	\lambda_{\phi(N)} \geq \frac{\pi\epsilon}{M-\frac{\epsilon}{2}}.
\end{equation}
%one can show that there exists a sequence 
%$B_{\phi(N)}\subset[-\pi,\pi]$ of Borel sets with Lebesgue measure larger than $\frac{\pi\epsilon}{M-\frac{\epsilon}{2}}$ (a quantity which does not depend on N) such that
Since $B_{\phi(N)}$ is defined such that,
\begin{equation}
	\label{eq=diffOfExpectedPers}
	\left|
		\overline{S}_{\Y}\sfN(\omega;\bgamma\ufN) - \overline{S}_{\Y}\sfN(\omega;\btheta)
	\right|
	\geq \frac{\epsilon}{2}, \ \forall\omega\in B_{\phi(N)},
\end{equation}
it follows that, dividing each side of~\eqref{eq=diffOfExpectedPers} by $\overline{S}_{\Y}\sfN(\omega;\bgamma\ufN)$,
\begin{equation}
	\left|
		\frac{\overline{S}_{\Y}\sfN(\omega;\btheta)}{\overline{S}_{\Y}\sfN(\omega;\bgamma\ufN)} - 1
	\right| \geq \frac{\epsilon}{2\overline{S}_{\Y}\sfN(\omega;\bgamma\ufN)}
	\geq \frac{\epsilon}{2\overline{S}_{\Y,\max}}, \ \forall\omega\in B_{\phi(N)}.
\end{equation}
We therefore have that for all $\omega\in B_{\phi(N)}$,
$\left|\overline{S}_{\Y}\sfN(\omega;\btheta)/\overline{S}_{\Y}\sfN(\omega;\bgamma\ufN)-1\right|$ is bounded below by $c\epsilon$, where $c=1/(2\overline{S}_{\Y,\max})$ is a positive constant independent of $N$. 
Denote 
\begin{equation*}
	b:x\rightarrow x-\log x-1, \ \ x>0, \ \ \ \ \  \overline{b}\sfN:\omega\mapsto b\left(\frac{\overline{S}_{\Y}\sfN(\omega;\btheta)}{\overline{S}_{\Y}\sfN(\omega;\bgamma\ufN)}\right), \ \ \omega\in[-\pi,\pi].
\end{equation*}
 For all $\omega\in B_{\phi(N)}$, $\overline{b}\sfN(\omega)$ is bounded below by $d=\min(b(1+c\epsilon,1-c\epsilon))$ (where $d>0$ is a constant that depends on $\epsilon$ but not on $N$) because of the properties of the function $b(x)$ which we recalled at the beginning of the proof of lemma \ref{lemma=uniquenessOfMin}.	
The function $b(x)$ has a bounded derivative on any interval of the form $[a_1,a_2]$ where $0<a_1<a_2<\infty$. Since 
\begin{equation*}
\frac{\overline{S}_{\Y,\text{max}}}{\overline{S}_{\Y,\text{min}}} \geq \frac{\overline{S}_{\Y}\sfN(\omega;\btheta)}{\overline{S}_{\Y}\sfN(\omega;\bgamma\ufN)} \geq  \frac{\overline{S}_{\Y,\text{min}}}{\overline{S}_{\Y,\text{max}}}>0, 
\end{equation*}
and using Lemma~\ref{lemma=boundexpectedperiodogram} and Lemma~\ref{lemma=totalvariationnorm}, the function $\overline{b}\sfN(\omega)$ has a bounded derivative. We denote the corresponding bound $l_{\max}$, which is independent of $N$.

%Let $\lambda_{\phi(n)}$ be the Lebesgue measure of the Borel set $B_{\phi(n)}$. One can show that there exist $T=\lfloor{\frac{N \lambda_{\phi(n)}}{4\pi} \rfloor$ non-increasing elements $\nu_1, \cdots, \nu_T\in B_{\phi(n)}$ such that $\nu_{i+1}-\nu_i \geq \frac{4\pi}{N}, i=1,\cdots,T-1$. Then there exist $T-1$ Fourier frequencies
%$\nu'_1, \cdots, \nu'_{T-1}$, such that $\nu_i < \nu'_i < \nu_{i+1}$. Then we have
%\begin{equation}
	%\sum_{i=1}^{T-1}{\left| 
	%b\left( \frac{\overline{S}_{\Y}\sfN(\nu'_i;\btheta)}{\overline{S}_{\Y}\sfN(\nu'_i;\bgamma_N)}\right) - 
	%b\left( \frac{\overline{S}_{\Y}\sfN(\nu_i;\btheta)}{\overline{S}_{\Y}\sfN(\nu_i;\bgamma_N)}\right)
	%\right|}
	%\leq V_b,
%\end{equation}
%which implies
%\begin{eqnarray}
	%\sum_{i=1}^{T-1}{b\left( \frac{\overline{S}_{\Y}\sfN(\nu'_i;\btheta)}{\overline{S}_{\Y}\sfN(\nu'_i;\bgamma_N)}\right)} &\geq &
	%\sum_{i=1}^{T-1}{ b\left( \frac{\overline{S}_{\Y}\sfN(\nu_i;\btheta)}{\overline{S}_{\Y}\sfN(\nu_i;\bgamma_N)}\right)} - V_b\\
	%&\geq& Td\epsilon-V_b.
%\end{eqnarray}
Recalling that $\lambda_{\phi(N)}$ is the measure of $B_{\phi(N)}$, there exist $T=\lfloor\frac{N \lambda_{\phi(N)}}{4\pi} \rfloor$ increasing elements $\nu_1, \cdots, \nu_T\in B_{\phi(N)}$ such that $\nu_{i+1}-\nu_i \geq \frac{4\pi}{N}, i=1,\cdots,T-1$. Then there exist $T-1$ Fourier frequencies
$\nu'_1, \cdots, \nu'_{T-1}$, such that $\nu_i < \nu'_i < \nu_{i+1}$. Then we have,
\begin{equation*}
	\left|\sum_{i=1}^{T-1}{ 
	\overline{b}\sfN(\nu'_i)-
	\overline{b}\sfN(\nu_i)
	}\right|
	\leq
	\sum_{i=1}^{T-1}{\left| 
	\overline{b}\sfN(\nu'_i)-
	\overline{b}\sfN(\nu_i)
	\right|}
	\leq \sum_{i=1}^{T-1}{(\nu'_i-\nu_i)l_{\max}} \leq 2\pi l_{\max},
\end{equation*}
which implies
\begin{eqnarray*}
	\sum_{i=1}^{T-1}{\overline{b}\sfN(\nu'_i)} &\geq &
	\sum_{i=1}^{T-1}{\overline{b}\sfN(\nu_i)} - 2\pi l_{\max}\\
	&\geq& (T-1)d-2\pi l_{\max}.
\end{eqnarray*}

Because $T$ is of order $N$, we conclude that
(\ref{eq=d-d}) cannot converge to zero. We arrive at a contradiction.\textreferencemark \linebreak
So we obtain that for all integer $\tau$, $\overline{c}_{\Y}\sN(\tau;\bgamma_N)-\overline{c}_{\Y}\sN(\tau;\btheta)$ converges to zero when $N$ goes to infinity.
\item
In particular for $\tau\in\Gamma$, if we denote $\alpha_\tau = \liminf\limits_{N\rightarrow\infty}\left|c_g\sN(\tau)\right| > 0$,  we have for $N$ large enough,
\begin{equation*}
	\left|\overline{c}_{\Y}\sN(\tau;\bgamma_N)-\overline{c}_{\Y}\sN(\tau;\btheta)\right| = \left|c_g\sN(\tau)\left(c_X(\tau;\bgamma_N)-c_X(\tau;\btheta)\right)\right| \geq \alpha_\tau \left|c_X(\tau;\bgamma_N)-c_X(\tau;\btheta)\right|,
\end{equation*}
so that $\left|c_X(\tau;\bgamma_N)-c_X(\tau;\btheta)\right|$ converges to zero as $N$ tends to infinity.
Because of the compacity of $\Theta$, and using the fact that the function $\btheta\mapsto\left\{c_X(\tau):\tau\in\Gamma\right\}$ is one-to-one and continuous, this yields the stated lemma.
\end{enumerate}
\end{proof}

%Let $T$ be a positive integer and let $-\pi\leq \omega_0, \cdots, \omega_T\leq\pi$, we therefore have,
%\begin{eqnarray}
	%\sum_{i=0}^{T-1}{\left| \overline{S}_{\Y}\sN(\omega_{i+1})-\overline{S}_{\Y}\sN(\omega_i) \right|} &=& \sum_{i=0}^{T-1}{\left| \int_{-\pi}^\pi{S_X(\omega_{i+1}-\lambda)S_g\sN(\lambda)d\lambda}-\int_{-\pi}^\pi{S_X(\omega_i-\lambda)S_g\sN(\lambda)d\lambda} \right|}\\
	%&=& \sum_{i=0}^{T-1}{\left| \int_{-\pi}^\pi{\left(S_X(\omega_{i+1}-\lambda)-S_X(\omega_i-\lambda)\right)S_g\sN(\lambda)d\lambda} \right|}\\
	%&\leq& \sum_{i=0}^{T-1}{ \int_{-\pi}^\pi{\left|\left(S_X(\omega_{i+1}-\lambda)-S_X(\omega_i-\lambda)\right)S_g\sN(\lambda)\right|d\lambda} }\\
	%&\leq& g_{\max}^2\sum_{i=0}^{T-1}{ \int_{-\pi}^\pi{\left|S_X(\omega_{i+1}-\lambda)-S_X(\omega_i-\lambda)\right|d\lambda} }\\
	%&\leq& g_{\max}^2 \int_{-\pi}^\pi{\sum_{i=0}^{T-1}{\left|S_X(\omega_{i+1}-\lambda)-S_X(\omega_i-\lambda)\right|d\lambda} }\\
	%&\leq& g_{\max}^2 V_X
%\end{eqnarray}
%This concludes the proof as the upper bound $V_X$ does not depend on the parameter $\btheta$.

\subsection{Proof of Lemma \ref{lemma=boundOnVariance}}
\label{proof=boundOnVariance}
\begin{proof}
	Let $a_\text{max}$ be a finite positive constant such that $\left|a\sN(\omega)\right|\leq a_\text{max}, \forall\omega\in[-\pi,\pi),\forall N\in\N$.
	We start by looking at the covariance matrix of the Fourier transform. We shall denote the Fourier transform, for a fixed positive integer $N$,
	\begin{equation*}
		J_{\Y}\sN(\omega) = \frac{1}{\sqrt{N}}\sum_{t=0}^{N-1}{\Y_te^{-i\omega t}} = \frac{1}{\sqrt{N}}\sum_{t=0}^{N-1}{g_t X_te^{-i\omega t}}, \ \ \omega\in\Omega_N.
	\end{equation*}
	Since the expectation of the latent process is assumed to be zero, the same holds for the Fourier transform by the linearity of the Fourier transform. Hence from the linear equation above we see that the covariance matrix elements can be expressed in the following way:
	\begin{equation}
		\label{eq=covFFT}
		\cov\left\{ J_{\Y}\sN(\omega), J_{\Y}\sN(\omega')\right\} = \frac{1}{N}\left(G_{\omega}\sN\right)^H C_X\sN(\btheta) G_{\omega'}\sN, \ \ \omega,\omega'\in\Omega_N,
	\end{equation}
	where subscript H denotes the Hermitian transpose, $C_X\sN(\btheta)$ denotes the finite theoretical autocovariance matrix of the latent process (i.e with elements $c_X(i-j;\btheta)$, $i,j=0,\cdots,N-1$), and 
	$G_\omega\sN$ is the vector $[g_t e^{i\omega t}: t=0,\cdots,N-1]^T$.
	Using Isserlis' theorem \citep{WhittleAdam} and the assumption of Gaussianity of the latent process (which in turns implies the Gaussianity of the Fourier transform of the modulated process), the covariances of the periodogram are related to the covariances of the Fourier transform according to the simple following relation
\begin{equation}
	\label{eq=covPeriodogram}
	\cov\left\{\hat{S}_{\Y}\sN(\omega), \hat{S}_{\Y}\sN(\omega')\right\} = \left| \cov\left\{ J_{\Y}\sN(\omega), J_{\Y}\sN(\omega')\right\} \right|^2.
\end{equation}
This can be written as
\begin{eqnarray*}
	\cov\left\{\hat{S}_{\Y}\sN(\omega), \hat{S}_{\Y}\sN(\omega')\right\} &=& \frac{1}{N^2} \left(G_{\omega}\sN\right)^H C_X\sN(\btheta) G_{\omega'}\sN \left( \left(G_{\omega}\sN\right)^H C_X\sN(\btheta) G_{\omega'}\sN \right)^H.\\
	&=& \frac{1}{N^2} {\left(G_{\omega}\sN\right)}^H C_X\sN(\btheta) G_{\omega'}\sN  {G_{\omega'}\sN}^H  {C_X\sN(\btheta)}^H G_{\omega}\sN.
\end{eqnarray*}
We then have
\begin{eqnarray*}
	&&\var\left\{\frac{1}{N}\sum_{\omega\in\Omega_N}{a\sN(\omega)\hat{S}_{\Y}\sN(\omega)}\right\}= \\
	&& \frac{1}{N^2}\sum_{\omega\in\Omega_N}{\sum_{\omega'\in\Omega_N}{a\sN(\omega)a\sN(\omega') \frac{1}{N^2}{\left(G_{\omega}\sN\right)}^H C_X\sN(\btheta) G_{\omega'}\sN  {G_{\omega'}\sN}^H  {C_X\sN(\btheta)}^H G_{\omega}\sN }}\\
	&\leq& \frac{a_{\text{max}}^2}{N^4}\sum_{\omega\in\Omega_N}{\sum_{\omega'\in\Omega_N}{ {\left(G_{\omega}\sN\right)}^H C_X\sN(\btheta) G_{\omega'}\sN  {G_{\omega'}\sN}^H  {C_X\sN(\btheta)}^H G_{\omega}\sN }}\\
	&=& \frac{a_{\text{max}}^2}{N^4}\sum_{\omega\in\Omega_N}{ {\left(G_{\omega}\sN\right)}^H C_X\sN(\btheta) \sum_{\omega'\in\Omega_N}\left\{G_{\omega'}\sN  {G_{\omega'}\sN}^H\right\}  {C_X\sN(\btheta)}^H G_{\omega}\sN },
\end{eqnarray*}
where the first inequality is legitimate as the covariances of the periodogram are positive real-valued numbers (see (\ref{eq=covPeriodogram})), and where the last equality is obtained after factorizing. Now we use the fact that
\begin{equation}
	\label{eq=matrixDiagonal}
	\sum_{\omega'\in\Omega_N}\left\{G_{\omega'}\sN  {G_{\omega'}\sN}^H\right\} = N \text{diag}(g_0^2, \cdots, g_{N-1}^2).
\end{equation}
Indeed, the $(t_1,t_2)$-th term of the left hand side of~\eqref{eq=matrixDiagonal} is given by
\begin{equation}
\sum_{\omega'\in\Omega_N}g_{t_1}g_{t_2}e^{i\omega'(t_1-t_2)}= \sum_{k=0}^{N-1}{g_{t_1}g_{t_2}e^{\frac{i2k\pi(t_1-t_2)}{N}}} = g_{t_1}g_{t_2}\sum_{k=0}^{N-1}{e^{\frac{i2k\pi(t_1-t_2)}{N}}},
\end{equation}
where we recognize the finite sum of the geometric sequence of term $e^{\frac{i2\pi(t_1-t_2)}{N}}$, which is $N$ if $t_1=t_2$, and otherwise,
\begin{equation}
	\sum_{k=0}^{N-1}{e^{\frac{i2k\pi(t_1-t_2)}{N}}} = \sum_{k=0}^{N-1}{\left(e^{\frac{i2\pi(t_1-t_2)}{N}}\right)^k} = \frac{1-\left(e^{\frac{i2\pi(t_1-t_2)}{N}}\right)^N}{1-e^{\frac{i2\pi(t_1-t_2)}{N}}}=0.
\end{equation}
Therefore
\begin{eqnarray}
	\var\left\{\frac{1}{N}\sum_{\omega\in\Omega_N}{a(\omega)\hat{S}_{\Y}\sN(\omega)}\right\}
	&\leq& \frac{a_{\text{max}}^2}{N^3}\sum_{\omega\in\Omega_N}{ \left(G_{\omega}\sN\right)^H C_X\sN(\btheta) \ \text{diag}(g_0^2, \cdots, g_{N-1}^2) \ C_X\sN(\btheta)^H G_{\omega}\sN }\\
	&\leq& \frac{a_{\text{max}}^2g_{\text{max}}^2}{N^3}\sum_{\omega\in\Omega_N}{ \left(G_{\omega}\sN\right)^H C_X\sN(\btheta)  {C_X\sN(\btheta)}^H G_{\omega}\sN }.
\end{eqnarray}
Therefore we now have
\begin{equation}
	\label{eq=variancefinale}
	\var\left\{\frac{1}{N}\sum_{\omega\in\Omega_N}{a(\omega)\hat{S}_{\Y}\sN(\omega)}\right\} \leq \frac{a_{\text{max}}^2g_{\text{max}}^2}{N^3}\sum_{\omega\in\Omega_N}{ \left\|{C_X\sN(\btheta)}^H G_{\omega}\sN\right\|_2^2},
\end{equation} 
where $\|\cdot\|_2$ denotes the Euclidean norm on $\C^N$. 
For all $\bold{U}\in\C^N$, the matrix $C_X\sN(\btheta)$ is Hermitian, so it can be written $PDP^H$ where D is a diagonal matrix and where $P$ is unitary, so that,
\begin{equation}
	\left\|C_X\sN(\btheta) \bold{U} \right\|_2  
	\leq 
	\left\|\bold{U}\right\|_2 \max_{\eta\in\text{sp}\left(C_X\sN(\btheta)\right)} |\eta|,
\end{equation}
where $\text{sp}\left(C_X\sN(\btheta)\right)$ is the set of eigenvalues of $C_X\sN(\btheta)$. Furthermore we have from \citet[p. 394]{matrixanalysis} that, recalling that the spectral density $S_X(\omega;\btheta)$ is assumed to be continuous in $\omega$,
\begin{equation}
	\label{eq=boundEigenvalues}
	\max_{\eta\in\text{sp}\left(C_X\sN(\btheta)\right)} |\eta| =\max_{\bold{U}\in\C^n}\left\{ \frac{\bold{U}^H C_X\sN(\btheta) \bold{U}}{\bold{U}^H \bold{U}}\right\} \leq S_{X,\max}.
\end{equation}
Combining~\eqref{eq=variancefinale}-\eqref{eq=boundEigenvalues} and replacing $\bold{U}$ by $G_{\omega}\sN$,
\begin{equation}
\var\left\{\frac{1}{N}\sum_{\omega\in\Omega_N}{a(\omega)\hat{S}_{\Y}\sN(\omega)}\right\} \leq \frac{\left(S_{X,\text{max}} \ a_{\max}\ g^2_{\max}\right)^2}{N},
\end{equation}
as $\left\|G_{\omega}\sN\right\|_2\leq g_{\max}\sqrt{N}$.
\end{proof}

\subsection{Proof of Lemma \ref{lemma=boundbelowC}}
\label{proof=boundbelowC}
\begin{proof}
We first show that the proposition is true for all vectors in $\mathcal{C} = \left\{V\in\R^d:\left\|V\right\|_2 = 1\right\}$, which is compact. The function $\mathcal{S}:V\mapsto\sum_{i=1}^d{\alpha_i^2\left(\bold{U}_i^T\bold{V}\right)^2}$ is continuous, so the image of $\mathcal{C}$ by $\mathcal{S}$ is compact. Since $\mathcal{S}$ takes non-negative values, the image of $\mathcal{C}$ by $\mathcal{S}$ either contains zero or there exists a constant $C>0$ such that it is bounded below by $C$. The image of $\mathcal{C}$ by $\mathcal{S}$ cannot contain zero, as otherwise there would be a vector of $\R^d$ with norm 1 whose scalar product with vectors $U_i, i=1,\cdots,d$ is zero, which is impossible as we have assumed that the family $U_1,\cdots,U_d$ has rank d. Therefore there exist a constant $C>0$ such that,
\begin{equation*}
		\sum_{i=1}^d{\alpha_i^2\left(\bold{U}_i^T\bold{V}\right)^2} \geq C,  \ \forall V\in\mathcal{C}.
\end{equation*}
Now in general, if $\bold{V}$ is any non-zero vector in $\R^d$, we have, using the result we have derived for vectors of $\mathcal{C}$,
	\begin{eqnarray*}
		\sum_{i=1}^d{\alpha_i^2\left(\bold{U}_i^T\bold{V}\right)^2} &=& \left\|\bold{V}\right\|_2^2\sum_{i=1}^d{\alpha_i^2\left(\bold{U}_i^T\frac{\bold{V}}{\left\|\bold{V}\right\|_2}\right)^2}\\
		&\geq&\left\|\bold{V}\right\|_2^2 C.
	\end{eqnarray*}
	If $V=0$ the result is obsvious.
	This concludes the proof in the general case.
\end{proof}

\subsection{Proof of Lemma \ref{lemma=7}}
\label{proof=7}
\begin{proof}
	\begin{enumerate}
		\item Direct calculations give that the score function takes the form, for $i=1,\cdots,d$,
		\begin{equation*}
			\frac{\partial l_M\sN}{\partial\theta_i}(\btheta) = \frac{1}{N}\sum_{\omega\in\Omega_N}\left\{ \frac{\frac{\partial \overline{S}_{\Y}\sN}{\partial \theta_i}(\omega;\btheta)}{\left(\overline{S}_{\Y}\sN(\omega;\btheta)\right)^2}\left( \overline{S}_{\Y}\sN(\omega;\btheta) - \hat{S}_{\Y}\sN(\omega) \right) \right\}.
		\end{equation*}
		Since by definition $\overline{S}_{\Y}\sN(\omega;\btheta) = \E\left\{ \hat{S}_{\Y}\sN(\omega);\btheta \right\}$, the null expectation of the score function is zero.
		Applying Lemma~\ref{lemma=boundOnVariance}, and the fact that $\frac{\frac{\partial \overline{S}_{\Y}\sN}{\partial \theta_i}(\omega;\btheta)}{\left(\overline{S}_{\Y}\sN(\omega;\btheta)\right)^2}$ is bounded above in absolute value independently of $N$ (as a direct consequence of Lemma~\ref{lemma=boundexpectedperiodogram} and Lemma~\ref{lemma=totalvariationnorm}), we get the first result.
		\item Again by direct calculation we obtain the following Hessian matrix:
		\begin{eqnarray*}
			\frac{\partial^2l_M\sN}{\partial \theta_i \partial\theta_j}(\btheta) &= &
			\frac{1}{N} \sum_{\omega\in\Omega_N}
			\left\{
			\frac{\frac{\partial^2\overline{S}_Y\sN}{\partial \theta_i \partial \theta_j}(\omega;\btheta)\left( \overline{S}_Y\sN(\omega;\btheta)\right)^2
			-2\overline{S}_Y\sN(\omega;\btheta)\frac{\partial\overline{S}_Y\sN}{\partial\theta_i}(\omega;\btheta)\frac{\partial\overline{S}_Y\sN}{\partial\theta_j}(\omega;\btheta)}
			{\left( \overline{S}_Y\sN(\omega;\btheta)\right)^4}\right. \\
			&\times&\left.
			\left( \overline{S}_Y\sN(\omega;\btheta) - \hat{S}_Y\sN(\omega) \right)
			+  
			\frac{1}{\left( \overline{S}_Y\sN(\omega;\btheta)\right)^2}
			\frac{\partial\overline{S}_Y\sN}{\partial\theta_i}(\omega;\btheta)
			\frac{\partial\overline{S}_Y\sN}{\partial\theta_j}(\omega;\btheta)
			\right\}.
		\end{eqnarray*}
		The expectation of the Hessian matrix is therefore
		\begin{equation*}
			\mathcal{I}\sN(\btheta) = \frac{\partial^2 l_M\sN}{\partial\btheta\partial\btheta^T}(\omega;\btheta) =  \frac{1}{N}\sum_{\omega\in\Omega_N}{\frac{1}{\left( \overline{S}_{\Y}\sN(\omega;\btheta)\right)^2}
			\frac{\partial\overline{S}_{\Y}\sN}{\partial\btheta}(\omega;\btheta)\frac{\partial\overline{S}_{\Y}\sN}{\partial\btheta^T}(\omega;\btheta)},
		\end{equation*}
		where we use the notation $\frac{\partial\overline{S}_{\Y}\sN}{\partial\btheta^T}(\omega;\btheta)$ to denote the transpose of the gradient vector $\frac{\partial\overline{S}_{\Y}\sN}{\partial\btheta}(\omega;\btheta)$.
	For any of the $\omega\in\Omega_N$ (to which corresponds a term in the above sum), and for any vector $\bold{U}\in\R^d$,
	\begin{equation*}
		\bold{U}^T \frac{\partial\overline{S}_{\Y}\sN}{\partial\btheta}(\omega;\btheta)\frac{\partial\overline{S}_{\Y}\sN}{\partial\btheta^T}(\omega;\btheta) \bold{U} = \left| \frac{\partial\overline{S}_{\Y}\sN}{\partial\btheta^T}(\omega;\btheta) \bold{U} \right|^2 \geq 0,
 	\end{equation*}
	so that the matrix $\mathcal{I}\sN(\btheta)$ is non-negative definite as a sum of non-negative definite matrices.
	Now to show that the matrix $\mathcal{I}(\btheta)$ is positive definite, let $\bold{U} = [u_1,\cdots,u_d]^T\in\R^d$ non-zero.
	We have
	\begin{eqnarray*}
		\bold{U}^T \ \mathcal{I}\sN(\btheta) \ \bold{U} &=& \frac{1}{N}\sum_{\omega\in\Omega_N}{\frac{1}{\left( \overline{S}_{\Y}\sN(\omega;\btheta)\right)^2}
			\bold{U}^T
			\frac{\partial\overline{S}_{\Y}\sN}{\partial\btheta}(\omega;\btheta)\frac{\partial\overline{S}_{\Y}\sN}{\partial\btheta^T}(\omega;\btheta)}
			\bold{U}\\
			&=& \frac{1}{N}\sum_{\omega\in\Omega_N}{\frac{1}{\left( \overline{S}_{\Y}\sN(\omega;\btheta)\right)^2}
			\left|\frac{\partial\overline{S}_{\Y}\sN}{\partial\btheta^T}(\omega;\btheta)
			\bold{U}
			\right|^2}\\
			&=& \frac{1}{N}\sum_{\omega\in\Omega_N}{\frac{1}{\left( \overline{S}_{\Y}\sN(\omega;\btheta)\right)^2}
			\left(
			\sum_{i=1}^d{
			\frac{\partial\overline{S}_{\Y}\sN}{\partial\theta_i}(\omega;\btheta)
			u_i
			}
			\right)^2
			}\\
			&\geq& \frac{1}{N\left( \sup_{\omega\in\Omega_N}\overline{S}_{\Y}\sN(\omega;\btheta)\right)^2}\sum_{\omega\in\Omega_N}{
			\left(
			\sum_{i=1}^d{
			\frac{\partial\overline{S}_{\Y}\sN}{\partial\theta_i}(\omega;\btheta)
			u_i
			}
			\right)^2.
			}
	\end{eqnarray*}
	Seeing $\ \left\{\frac{1}{\sqrt{N}}\sum_{i=1}^d{u_i\frac{\partial\overline{S}_{\Y}\sN}{\partial\theta_i}(\omega;\btheta)}\right\}_{\omega\in\Omega_N}
	\Leftrightarrow \ 
	\left\{\sum_{i=1}^d{u_i\frac{\partial\overline{c}_{\Y}\sN}{\partial\theta_i}(\tau;\btheta)}\right\}_{\tau=-(N-1), \cdots, N-1}$
	as a finite Fourier pair and applying Parseval's equality
	we obtain that,
	\begin{eqnarray*}
		\bold{U}^T \ \mathcal{I}\sN(\btheta) \ \bold{U} &\geq& 
		\frac{1}{\left( \sup_{\omega\in\Omega_N}\overline{S}_{\Y}\sN(\omega;\btheta)\right)^2}
		\sum_{\tau=-(N-1)}^{N-1}{
		\left(
		\sum_{i=1}^d{
		u_i
			\frac{\partial\overline{c}_{\Y}\sN}{\partial\theta_i}(\tau;\btheta)
			}
			\right)^2
			}\\
		&=&
		\frac{1}{\left( \sup_{\omega\in\Omega_N}\overline{S}_{\Y}\sN(\omega;\btheta)\right)^2}
		\sum_{\tau=-(N-1)}^{N-1}{
		\left(
		\sum_{i=1}^d{
		u_i
			c_g\sN(\tau)\frac{\partial c_X}{\partial\theta_i}(\tau;\btheta)
			}
			\right)^2,
			}\\
	\end{eqnarray*}
	by definition of $\overline{c}_{\Y}\sN(\tau;\btheta)$ and noting that $c_g\sN(\tau)$ does not depend on $\btheta$.
	Therefore,
	\begin{equation*}
			\bold{U}^T \ \mathcal{I}\sN(\btheta) \ \bold{U} \geq
		\frac{1}{\left( \sup_{\omega\in\Omega_N}\overline{S}_{\Y}\sN(\omega;\btheta)\right)^2}
		\sum_{\tau\in\Gamma}{
		\left(
		\sum_{i=1}^d{
		u_i
			c_g\sN(\tau)\frac{\partial c_X}{\partial\theta_i}(\tau;\btheta)
			}
			\right)^2
			},
	\end{equation*}
	as long as $N$ is larger than the greater integer value in $\Gamma$.
	Denote $\alpha_{\tau} = \liminf\limits_{N\rightarrow\infty}\left|c_g\sN(\tau)\right| > 0$ for $\tau\in\Gamma$, we obtain that for $N$ large enough (see~\eqref{eq=equivalenceoflimitinf}),
	\begin{eqnarray*}
	\bold{U}^T \ \mathcal{I}\sN(\btheta) \ \bold{U} &\geq& 
		\frac{1}{\left( \sup_{\omega\in\Omega_N}\overline{S}_{\Y}\sN(\omega;\btheta)\right)^2}
		\sum_{\tau\in\Gamma}{
		c_g\sN(\tau)^2
		\left(
		\sum_{i=1}^d{
		u_i
			\frac{\partial c_X}{\partial\theta_i}(\tau;\btheta)
			}
			\right)^2
			}\\
			& \geq & 
		\frac{1}{\left( \sup_{\omega\in\Omega_N}\overline{S}_{\Y}\sN(\omega;\btheta)\right)^2}
		\sum_{\tau\in\Gamma}{
		\alpha_\tau^2
		\left(
		\sum_{i=1}^d{
		u_i
			\frac{\partial c_X}{\partial\theta_i}(\tau;\btheta)
			}
			\right)^2
			}.
	\end{eqnarray*}
	Now according to the assumption of significant correlation contribution, the mapping $\btheta\mapsto \left[c_X(\tau):\tau\in\Gamma\right]^T$ is one-to-one, so its Jacobian taken at the true parameter vector $\btheta$ is non-zero. Therefore the family $\frac{\partial c_X(\tau)}{\partial \theta}:\tau\in\Gamma$ has rank $d$ and we can apply Lemma~\ref{lemma=boundbelowC}, i.e. we can conclude that
there exists a positive constant $C>0$ such that,
	\begin{equation*}
	\bold{U}^T \ \mathcal{I}\sN(\btheta) \ \bold{U} \geq C \left\|U\right\|_2^2.
	\end{equation*}
This implies that the norm of the expected Hessian matrix is bounded below by a positive real-number.
Similarly to the gradient, using Lemma~\ref{lemma=boundOnVariance}, we obtain the stated result for the Hessian.
\end{enumerate}
\end{proof}

\end{document}